\documentclass[reqno,11pt]{article}
\pdfoutput=1
\usepackage{jheppub}
\usepackage[T1]{fontenc}
\usepackage[latin9]{inputenc}
\usepackage{amsmath}
\usepackage{amsthm}
\usepackage{amssymb}
\usepackage{cancel}
\usepackage{stmaryrd}
\usepackage[dvipsnames]{xcolor}
\usepackage[normalem]{ulem}
\usepackage[tikz]{bclogo}
\usetikzlibrary{shapes, shadows, arrows, patterns, shadings} 
\usetikzlibrary{positioning}
\tikzset{mynode/.style={align=center,text width=3cm}}
\usepackage{pgf}
\usepackage{hyperref}
\usepackage{slashed}
\usepackage{shuffle}
\usepackage{babel}
\usepackage{graphicx}
\usepackage{dcolumn}
\usepackage{bm}
\usepackage{bbm}
\usepackage[scr,scaled=1.0]{rsfso}

\newcommand{\CC}{\ensuremath{\mathbb{C}}}
\newcommand{\ZZ}{\ensuremath{\mathbb{Z}}}

\DeclareMathOperator{\Sol}{Sol}
\DeclareMathOperator{\rank}{rank}

\newcommand{\dd}{\mathrm{d}}

\newtheorem{theorem}{Theorem}[section]

\newtheorem{lemma}[theorem]{Lemma}

\title{Canonical differential equations for Feynman integrals from 
$\mathcal{A}$-hypergeometric systems in the Schwinger representation}

\author{Mateo Jimenez-Santacruz, Cristhiam Lopez-Arcos, Alexander Quintero V\'{e}lez}
\affiliation{Departamento de Matem\'{a}ticas, Universidad Nacional de Colombia Sede Medell\'{i}n,\\
Carrera 65 $\#$ 59A--110, Medell\'{i}n, Colombia}%
\emailAdd{mjimenezsa@unal.edu.co, cmlopeza@unal.edu.co, aquinte2@unal.edu.co}

\date{\today}

\abstract{We follow a $D$-module approach to the construction of canonical differential equations for Feynman integrals directly from their Gel'fand-Kapranov-Zelevinsky (GKZ) hypergeometric systems. Starting from the Schwinger representation, we identify the associated generalized Euler integral and its GKZ system, as an alternative to the Lee-Pomeransky representation. The generalized Euler formulation allows reductions associated with facets of the Newton polytope to be identified directly, reducing the number of variables in the resulting differential systems. We construct the associated Pfaffian systems using Frobenius bases, which also provide direct access to the singular loci. After a suitable rationalizing change of variables, the systems are brought into canonical $\epsilon$-form. A key feature of the construction is that the dimension of the canonical basis is determined by the holonomic rank of the GKZ system and, for the examples considered, it is smaller than or equal to the number of master integrals obtained from integration-by-parts, thus providing a reduced description of the differential system.}

\begin{document}

\maketitle
\tableofcontents{}

\section{Introduction}
In perturbative quantum field theory, Feynman integrals play a crucial role: they constitute the building blocks of scattering amplitudes and correlation functions allowing high-precision theoretical predictions for experiments at the Large Hadron Collider, and more recently have found a place in the study of gravitational waves and cosmology \cite{Kosower:2022yvp,Buonanno:2022pgc,Driesse:2026qiz,Arkani-Hamed:2023kig,De:2023xue}. They are multidimensional integrals over loop momenta that depend on the external kinematical data of the process and the dimension of the spacetime. Despite their compact definition, their evaluation is often highly nontrivial, particularly at higher loop orders, due to the singularities and complex analytic structures that arise. Consequently, the analytic or numerical computation of these integrals has been the subject of extensive research and remains a central endeavor in modern theoretical and particle physics. 

Feynman integrals have also served as a bridge between theoretical physics and pure mathematics. Initially, as formal devices for organizing perturbative quantum field theory calculations, they have revealed themselves over the past two decades to be a rich source of genuinely new mathematics. Their singularity structure is governed by algebraic geometry through Landau varieties \cite{Mizera:2021icv,Fevola:2023kaw}, Newton polytopes and GKZ hypergeometric systems \cite{Klausen:2019hrg,delaCruz:2019skx}, while their transcendental content is related to the theory of periods, mixed Hodge structures and motives \cite{Bogner:2007mn,Bloch:2008jk,Vanhove:2014wqa,Marcolli:2009zy}. Other key developments concern analytical tools to solve their systems of differential equations \cite{Bitoun:2017nre,Agostini:2022cgv,Lairez:2022zkj,delaCruz:2024vwf,Rodriguez:2025bny}.

Among the algebraic and analytic methods that have been developed, one of the most widely used approaches for evaluating Feynman integrals is a combination of Integration by Parts (IBP) identities \cite{Tkachov:1981wb,Chetyrkin:1981qh,Laporta:2000dsw} and differential equations for master integrals \cite{Henn:2013pwa}. These two methods have made many high-precision computations possible, but persistent challenges remain: the function spaces needed to describe their solution can be difficult to identify and organize \cite{Britto:2026qjn} and the number of master integrals, and hence the number of differential equations, grows rapidly with the number of loops and legs of the diagrams \cite{Laporta:2000dsw}. A representative example of this increasing complexity is provided by the QED corrections to the electron anomalous magnetic moment. After integration-by-parts reduction, the one-, two-, three-, and four-loop contributions involve 1, 3, 17, and 334 master integrals, respectively \cite{Laporta:1996mq,LAPORTA2017232}.

The mentioned challenges have motivated the development of complementary mathematical frameworks aimed at uncovering the algebraic, analytical and geometric structures underlying Feynman integrals. Among them, the theory of $A$-hypergeometric systems, introduced by Gel'fand, Kapranov, and Zelevinsky (GKZ) \cite{GelfandGraevZelevinsky1987,GelfandKapranovZelevinsky1989,GelfandKapranovZelevinsky1990}, provides a natural setting for studying their parametric representations. In particular, the Lee-Pomeransky (LP) representation, based on the polynomial $G=\mathcal{U}+\mathcal{F}$, where $\mathcal{U}$ and $\mathcal{F}$ are the first and second Symanzik polynomials, respectively, can be embedded into a GKZ system by promoting the coefficients of $G$ to independent variables \cite{delaCruz:2019skx,Klausen:2019hrg}. This framework allows one to compute series solutions for the Feynman integral using the Saito-Sturmfels-Takayama canonical series algorithm \cite{SaitoSturmfelsTakayama2000} or regular unimodular triangulations associated with the GKZ data. In many examples, this construction can be formulated as a tractable problem in computational algebra.
    
Nevertheless, obtaining closed-form representations of these series in terms of known hypergeometric functions can be challenging. Even when such representations are available, computing their $\epsilon$-expansion is often difficult, particularly in the multivariate case; see, for example,  \cite{Moch:2001zr,Kalmykov:2008ge}. In addition, the holonomic rank of the generic GKZ system may exceed the number of master integrals obtained after IBP reduction. Although recent methods exploit resonance to reduce GKZ systems to smaller ones, systematically restricting the enlarged coefficient space to the physical locus remains a significant challenge \cite{Vanhove:2018mto,delaCruz:2019skx,Britto:2026qjn}.

Another important mathematical incursion comes from the application of $D$-modules techniques as a complementary algebraic framework for the calculation of Feynman integrals \cite{Henn:2023tbo,Chestnov:2025whi}. This point of view focuses on the questions related to differential equations, singular loci, and the number of independent solutions to be analyzed in terms of the algebraic and geometric properties of $D$-ideals and Pfaffian systems.  

In this work we also adopt a GKZ/$D$-modules perspective. Instead of constructing the system directly from the LP polynomial $G$, we begin with the Schwinger-parametric representation. We introduce a change of variables which yields a generalized Euler integral in which the two transformed Symanzik polynomials remain as separate polynomial blocks. By promoting the coefficients of their monomials to independent variables, the resulting integral can be embedded into a GKZ system associated with a Cayley-type configuration. The two-block structure preserves the distinct homogeneity properties and coefficient sets of the transformed Symanzik polynomials. This structure, in some cases, enabled us to identify admissible facet reductions of the Newton polytope directly, reducing the number of toric variables of the GKZ system. 

For each reduced system, instead of looking for the $\Gamma$-series solutions, we use the holonomic rank to determine the dimension of the corresponding Frobenius basis and their corresponding Pfaffian systems, which made the singular loci of the integrals explicit. Then, after suitable rationalizing changes of variables and basis transformations, these systems can be brought into canonical $\epsilon$-form \cite{Henn:2013pwa} and be solved as iterated Chen integrals, once appropriate boundary conditions are specified. We illustrated this procedure for the two-mass bubble, the on-shell massless box, the off-shell massless triangle and box, and the sunset integral with one massless propagator. One important remark, is that the holonomic rank of the resulting systems was checked against master-integral counts obtained with \textsc{LiteRed}. In all examples considered, the holonomic rank is less than or equal to the number of master integrals obtained through a conventional IBP reduction.

This work is organized as follows. In Section \ref{sec:back} we introduce the formal background material supporting our approach, namely $D$-modules, GKZ systems, and their connections. Section \ref{sec:newApproach} is where we present the connection between GKZ systems and Feynman integrals in the Schwinger representation. We present a series of examples of application in Section \ref{sec:examples}. We conclude in Section \ref{sec:concl}. Appendix \ref{app:A} is devoted to the proof of the \emph{Pfaffian reduction lemma} that is originally stated at the end of Section \ref{sec:2.2}. In Appendix \ref{app:bund-cond} we analyze the threshold boundary conditions for the bubble and on-shell box examples. In Appendix \ref{app:explicitGM} we show the explicit expressions for the Pfaffian matrices for the massless triangle/box and sunset examples.

\section{Background material}\label{sec:back}
This section collects the material needed in the remainder of the paper. We briefly review the basic notions of $D$-ideals and Pfaffian systems, introduce GKZ hypergeometric systems, and discuss their generalized Euler integral representations, while fixing the notation and conventions adopted throughout.

\subsection{$D$-ideals and Pfaffian systems}
We begin with $D$-ideals and the associated Pfaffian systems. A detailed treatment is available in \cite{SaitoSturmfelsTakayama2000} and \cite{SattelbergerSturmfels2025}. For an application in the context of Feynman integrals, we refer to e.g. \cite{Henn:2023tbo,Chestnov:2025whi}.

Let $\CC(x_1,\dots,x_n)$ denote the field of rational functions in the variables $x_1,\dots,x_n$. Given a rational function $a(x)=a(x_1,\dots,x_n)$, we define the multiplication of the differential operators $\partial_i$ and $a(x)$ by
\begin{equation}
\partial_i a(x) = a(x) \partial_i + \frac{\partial a(x)}{\partial x_i}. 
\end{equation}
We also require that $\partial_i$ and $\partial_j$ commute. These relations, together with the associative and distributive laws, define the ring of differential operators with rational function coefficients,
\begin{equation}
R_n = \CC(x_1,\dots,x_n) \langle \partial_1, \dots, \partial_n \rangle. 
\end{equation}
Every element $P \in R_n$ can be expressed in the form 
\begin{equation}
P = \sum_{\alpha \in \ZZ_{\geq 0}^{n}} a_{\alpha}(x) \partial^{\alpha},
\end{equation}
where only finitely many coefficients $a_{\alpha}(x) \in \CC(x_1,\dots,x_n)$ are nonzero. Here, $\partial^{\alpha} = \partial_1^{\alpha_1} \cdots \partial_n^{\alpha_n}$ is the standard multi-index notation. To distinguish the action of $R_n$ on rational functions from multiplication in $R_n$, it is customary to denote the former by $\bullet$, so that
\begin{equation}
\partial^{\alpha} \bullet f = \frac{\partial^{\alpha_1 + \cdots + \alpha_n} f}{\partial x_1^{\alpha_1} \cdots \partial x_n^{\alpha_n} }. 
\end{equation}
One readily verifies that $(PQ) \bullet f = P \bullet (Q \bullet f)$ holds for all $P,Q \in R_n$. 

The ring $R_n$ provides an algebraic framework for the study of systems of linear partial differential equations. These systems are naturally encoded by left ideals of $R_n$, which are simply called $D$-\emph{ideals}. Indeed, given one such $D$-ideal $I \subseteq R_n$, Bernstein's theorem \cite{Bernstein1971} implies that $I$ admits a finite set of generators $P_1, \dots,P_r \in R_n$. The corresponding system of linear partial differential equations for the unknown function $f$ is then given by
\begin{equation}
P_i \bullet f = 0, \quad i=1,\dots,r.
\end{equation}
We denote its space of solutions by $\Sol(I)$. This space turns out to be completely determined by the quotient $R_n/I$ regarded as a left $R_n$-module. In fact, if the solutions are sought in some left $R_n$-module $M$ of functions, then $\Sol(I)$ is in canonical bijection with the set of $R_n$-module homomorphisms from $R_n/I$ to $M$. For this reason, one usually shifts the focus from the solution space to the quotient $R_n/I$, which is simply referred to as the $D$-\emph{module} associated with $I$. 

Another way to interpret the condition of finite holonomic rank is through the possibility of rewriting the system of partial differential equations encoded by $I$ as a finite-dimensional system of first-order differential equations. Let $r = \rank(I)$ and choose a basis $\{s_1,s_2,\dots,s_r \}$ of $R_n/I$ as a vector space over $\CC(x_1,\dots,x_n)$ such that $s_1 = 1$. We further assume that this basis is represented by monomials in the differential operators $\partial_i$. In practice, such a basis can be obtained from a Gr\"{o}bner basis with respect to a fixed term order, but we shall not need this construction here and simply assume that such a basis has been chosen. Since $R_n/I$ is a left $R_n$-module, the action of the differential operators $\partial_i$ on the chosen basis can be expressed, for $i = 1,\dots,n$ and $j = 1,\dots,r$, as 
\begin{equation}
\partial_i s_j =  \sum_{k = 1}^{r} p_{ij}^{k} s_k, 
\end{equation}
for some coefficients $p_{ij}^{k} \in \CC(x_1,\dots,x_n)$. In other words, the action of each differential operator $\partial_i$ on $R_n/I$ is represented by $r \times r$ matrix $\Omega_i$ with entries in $\CC(x_1,\dots,x_n)$, whose $(k,j)$-th entry is given by $p_{ij}^{k} $. We next consider a solution $f \in \Sol(I)$ and associate to it the vector 
\begin{equation}\label{eq:2.7}
F = \begin{pmatrix} s_1 \bullet f \\ s_2 \bullet f \\ \vdots \\ s_r \bullet f \end{pmatrix}. 
\end{equation}
Our choice of $s_1 = 1$ ensures that the first component of $F$ is the function $f$ itself. Applying the differential operators $\partial_i$ to the components of $F$ and using the representation above, we find that the following identities hold:
\begin{equation}\label{eq:2.8}
\frac{\partial F}{\partial x_i} = P_i F,  \quad i = 1,\dots,n.
\end{equation}
This system of first-order differential equations is called a \emph{Pfaffian system} associated with $I$. The compatibility of this system imposes an integrability condition on the matrices $P_i$, which takes the form
\begin{equation}
\frac{\partial P_j}{\partial x_i} -\frac{\partial P_i}{\partial x_j}  =   P_i P_j  - P_j P_i , 
\end{equation}
for all $i,j = 1,\dots, n$. This integrability condition admits a natural geometric interpretation. Introducing the matrix-valued one-form 
\begin{equation}
P = \sum_{i = 1}^{n} P_i dx_i, 
\end{equation}
the Pfaffian system can be written compactly as
\begin{equation}
d F = P F. 
\end{equation}
The integrability condition above is precisely the condition that the connection defined by $P$ is flat. In this way, the condition of finite holonomic rank allows the corresponding system to be interpreted as a system of horizontal sections of a flat connection. In geometric settings, this connection is known as the Gauss--Manin connection (see, e.g., \cite[Ch.~12]{arnold1988singularities2}).

\subsection{GKZ hypergeometric systems}\label{sec:2.2}
We now recall the definition of GKZ hypergeometric systems. For a comprehensive treatment, we refer the reader to \cite{GelfandGraevZelevinsky1987,GelfandKapranovZelevinsky1989,GelfandKapranovZelevinsky2008}. The references \cite{SaitoSturmfelsTakayama2000,SattelbergerSturmfels2025}, already cited in the previous section, will also be useful here.

A GKZ hypergeometric system is determined by a finite set $\mathcal{A} = \{ a_1, \dots, a_N \}$ in a lattice $\ZZ^{m}$ together with a parameter vector $\beta = (\beta_1,\dots,\beta_m) \in \CC^{m}$. We assume that $\mathcal{A}$ generates $\ZZ^{m}$ as an abelian group and that there exists a linear form $h \colon \ZZ^{m} \to \ZZ$ satisfying $h(a_i)=1$ for every $i = 1,\dots,N$. Writing the vectors $a_i$ as columns produces an $m \times N$ matrix of maximal rank,
\begin{equation}
A =\begin{pmatrix}
a_{11} & a_{12} & \cdots & a_{1N} \\
a_{21} & a_{22} & \cdots & a_{2N} \\
\vdots & \vdots & \ddots & \vdots \\
a_{m1} & a_{m2} & \cdots & a_{mN}
\end{pmatrix},
\end{equation}
where $a_{ij}$ denotes the $i$-th coordinate of the vector $a_j$. This matrix is usually referred to as the $\mathcal{A}$-\emph{matrix}. We also consider the lattice $L$ of integer relations among the elements of $\mathcal{A}$,
\begin{equation}
L = \{ l = (l_1,\dots,l_N) \in \ZZ^{N} \mid l_1 a_1 + \cdots + l_N a_N = 0 \}.
\end{equation}
This lattice is precisely the kernel of the $\mathcal{A}$-matrix $A$, regarded as a homomorphism from $\ZZ^N$ to $\ZZ^{m}$. The GKZ hypergeometric equations are a system of partial differential equations for an unknown function $\Phi = \Phi(z_1,\dots,z_N)$. This set consists in two groups. The first are the structure equations
\begin{equation}\label{eq:2.15}
\prod_{l_i > 0} \left( \frac{\partial}{\partial z_i}\right)^{l_i} \Phi - \prod_{l_i < 0} \left( \frac{\partial}{\partial z_i}\right)^{-l_i} \Phi =0,
\end{equation}
for every $l = (l_1,\dots,l_N) \in L$.  The second group consists of the homogeneity or Euler equations, 
\begin{equation}\label{eq:2.16}
\sum_{j=1}^{N} a_{ij} z_j \frac{\partial \Phi}{\partial z_j} - \beta_i \Phi = 0,
\end{equation}
for every $i = 1,\dots,m$. It is worth noting that the Euler equations \eqref{eq:2.16} can be integrated explicitly. Consequently, the solutions of the full GKZ hypergeometric system \eqref{eq:2.15}--\eqref{eq:2.16} may be described in terms of functions of $N-m$ independent variables. The corresponding reduced system, however, is generally more complicated than the original one.

We now reformulate the GKZ hypergeometric equations in the language of $D$-ideals reviewed in the previous subsection. To this end, we consider the ring of differential operators with rational function coefficients in the variables $z_1,\dots,z_N$, 
\begin{equation}
R_N = \CC(z_1,\dots,z_N) \langle \partial_1, \dots, \partial_N \rangle. 
\end{equation}
For each $l \in L$, we define the so-called ``box'' operator
\begin{equation}
\Box_{l} = \prod_{l_i > 0} \partial_i^{l_i} - \prod_{l_i < 0}  \partial_i^{-l_i},
\end{equation} 
and, for each $i=1,\dots,m$, we define the first-order differential operator
\begin{equation}
Z_i = \sum_{j=1}^{N}a_{ij}z_j\partial_j-\beta_i.
\end{equation}
In terms of these, the GKZ hypergeometric system \eqref{eq:2.15}--\eqref{eq:2.16} may be rewritten as
\begin{equation} \label{eq:2.20}
\begin{aligned}
\Box_{l}\bullet\Phi&=0, &&l\in L,\\
Z_i\bullet\Phi&=0, &&i=1,\dots,m.
\end{aligned}
\end{equation}
We may therefore consider the $D$-ideal $H_{\mathcal{A}}(\beta)$ generated by all operators $\Box_{l}$ with $l \in L$ together with the operators $Z_i$ for $i = 1,\dots,m$. This ideal is called the \emph{GKZ hypergeometric ideal} associated with the data $\mathcal{A}$ and $\beta$. 

Let us mention two important properties of the GKZ hypergeometric system in connection with the GKZ hypergeometric ideal $H_{\mathscr{A}}(\beta)$ that will be important in what follows. The first concerns the number of solutions of the system. To state it, we first introduce some terminology. Let $\Delta_{\mathcal{A}}$ denote the $m$-dimensional polytope obtained as the convex hull of the finite set $\mathcal{A}$. A parameter vector $\beta \in \CC^{m}$ is said to be \emph{non-resonant} if it does not lie on any integral translate of an affine hyperplane supporting a facet of  $\Delta_{\mathcal{A}}$. In such a case, it is known that the holonomic rank of the GKZ hypergeometric ideal $H_{\mathcal{A}}(\beta)$ coincides with the number of linearly independent solutions of the GKZ hypergeometric system. Moreover, this number is equal to the normalized volume of the polytope $\Delta_{\mathcal{A}}$, where ``normalized'' means $m!$ times the usual Euclidean volume. The second concerns the algebraic structure of the GKZ hypergeometric ideal $H_{\mathcal{A}}(\beta)$. When $\beta$ is non-resonant, one also has that $H_{\mathcal{A}}(\beta)$ is irreducible. Roughly speaking, this means that the corresponding GKZ hypergeometric system cannot be decomposed into simpler independent subsystems. 

It is perhaps worth taking a closer look at what happens in the complementary resonant case. Let $F$ be a face of the polytope $\Delta_{\mathcal{A}}$ with respect to which the parameter vector $\beta$ is resonant. Denote by $\mathcal{A}_F$ the subconfiguration consisting of those vectors of $\mathcal{A}$ lying on $F$, and by $\beta_F$ the parameter induced on the linear span of $F$. The associated GKZ hypergeometric ideal $H_{\mathcal{A}_F}(\beta_F)$ describes a lower-dimensional hypergeometric system attached to the face $F$. One of the remarkable features of the resonant case is that the global GKZ hypergeometric ideal $H_{\mathcal{A}}(\beta)$ already encodes this lower-dimensional system. Indeed, restricting the GKZ system to the face $F$ amounts algebraically to adjoining the variables corresponding to columns outside $F$ to the ideal $H_{\mathcal{A}}(\beta)$ and then intersecting with the ring of differential operators with rational function coefficients generated by the variables indexed by $F$. This construction produces precisely the GKZ hypergeometric ideal $H_{\mathcal{A}_F}(\beta_F)$. In practice, the analysis of the resonant case is therefore reduced to the study of this lower-dimensional GKZ hypergeometric system.

We conclude this section by describing a way of constructing the Pfaffian system associated with the GKZ hypergeometric ideal $H_{\mathcal{A}}(\beta)$. The construction follows the general framework discussed in the previous subsection, specialized to the hypergeometric setting. We begin by fixing a basis $\{l_1,\dots,l_n \}$ of the lattice of relations $L$, where $n = N-m$. Writing each basis element as $l_i = (l_{i1},\dots,l_{iN})$ we introduce the corresponding independent ``toric'' variables\footnote{These variables are called ``toric'' because they are invariant under the algebraic torus action naturally associated with the $\mathcal{A}$-matrix $A$.} 
\begin{equation}
x_i=z_1^{l_{i1}}\cdots z_N^{l_{iN}},
\quad
i=1,\dots,n.
\end{equation}
The Pfaffian system will be formulated in terms of these variables. We therefore consider the ring of differential operators with rational function coefficients in $x_1,\dots,x_n$,
\begin{equation}
\tilde{R}_n = \CC(x_1,\dots,x_n)\langle \theta_1,\dots, \theta_n \rangle,
\end{equation}
where the action of the operators $\theta_i$ on a rational function is given by
\begin{equation}
\theta_i \bullet f = x_{i} \frac{\partial f}{\partial x_i}.
\end{equation}
The next step consists in separating the monomial behavior imposed by the Euler equations from the remaining dependence on the toric variables. To this end, we consider solutions of the form
\begin{equation}\label{eq:2.24}
\Phi(z_1,\dots,z_N) = z_1^{\alpha_1}\cdots z_N^{\alpha_N}f(x_1,\dots,x_n),
\end{equation}
where $\alpha = (\alpha_1,\dots,\alpha_N)$ is chosen so that $A \alpha = \beta$. Substituting this ansatz into the GKZ hypergeometric system \eqref{eq:2.20} requires expressing the differential operators appearing in the first-order differential operators $Z_i$ in terms of the toric variables. This is accomplished by applying the chain rule, under which each logarithmic derivative $z_j \partial_j$ is replaced by the operator $\sum_{i=1}^{n} l_{ij} \theta_i + \alpha_j$ for $j = 1,\dots,N$. In this way, every box operator $\Box_l$ of the original GKZ hypergeometric system is transformed into a differential operator belonging to the ring $\tilde{R}_n$. We denote by $I_{\mathcal{A}}(\beta) \subseteq \tilde{R}_n$ the left ideal generated by all these transformed box operators. Since the Euler equations have already been absorbed into the monomial factor $z_1^{\alpha_1}\cdots z_N^{\alpha_N}$, the quotient $\tilde{R}_n / I_{\mathcal{A}}(\beta)$ encodes the remaining differential structure of the GKZ hypergeometric system in the toric variables. A fundamental result, whose formal proof is provided in Appendix~\ref{app:A}, asserts that this quotient is a finite-dimensional vector space over the field $\mathbb{C}(x_1,\dots,x_n)$. Moreover, its dimension coincides with the holonomic rank of the original GKZ hypergeometric ideal $H_{\mathcal{A}}(\beta)$. The construction now proceeds exactly as in the general case discussed in the previous subsection. We choose a basis $\{s_1,s_2,\dots, s_r \}$ of the vector space $\tilde{R}_n / I_{\mathcal{A}}(\beta)$, where $r = \rank(H_{\mathcal{A}}(\beta))$, represented by monomials in the differential operators $\theta_i$, and assume without loss of generality that $s_1=1$. The vector $F$ is then defined as in \eqref{eq:2.8}, with $f$ given by the toric factor in \eqref{eq:2.24}. Since the operators $\theta_i$ act naturally on the quotient $\tilde{R}_n / I_{\mathcal{A}}(\beta)$, their action on the chosen basis is represented by $r \times r$ matrices $\Omega_i$ with entries in $\CC(x_1,\dots,x_n)$. Consequently, one obtains the Pfaffian system
\begin{equation}\label{eq:2.13}
x_i \frac{\partial F}{\partial x_i} = \tilde{P}_i F,  \quad i = 1,\dots,n.
\end{equation}
In practice, however, while system \eqref{eq:2.13} naturally inherits the toric invariance of the underlying $D$-module, for us it will be computationally advantageous to convert it into a system written in terms of standard partial derivatives
\begin{equation}
 \frac{\partial F}{\partial x_i} = P_i F, \quad i =1,\dots,n,
\end{equation}
where each matrix $P_i$ algebraically absorbs the factor $1/x_i$ directly into its entries. Throughout the explicit examples considered in this work, we adopt this formulation by carrying out this entrywise absorption. We prefer this algebraic choice for three main reasons, related to structure and to algorithmic efficiency. First, standard symbolic algorithms for integrating Feynman integrals and multiloop amplitudes require first-order linear systems written explicitly in terms of standard partial derivatives $\partial_i$. Second, absorbing $1/x_i$ into $P_i$ yields exact algebraic simplifications modulo $I_{\mathcal{A}}(\beta)$, expressing all matrix entries in standard ``$\dd \log$'' form over the complete singularity locus. This reveals the full toric alphabet directly and isolates the physical kinematic thresholds, instead of hiding them behind the boundary divisors $x_i=0$. Third, this entrywise absorption frequently cancels high-degree polynomial factors that appear in the numerators as a result of high-order GKZ reductions. This lowers the overall algebraic degree of the rational coefficients, which in turn simplifies both the zero-curvature conditions and the integrability checks performed in computer algebra systems. 
\subsection{Generalized Euler integral representations}
We now describe GKZ hypergeometric systems associated with generalized Euler integrals. These systems provide the natural framework for the Feynman integrals studied later in this article. Our exposition is based on \cite{MatsubaraHeo2020I}. For a complementary perspective, the reader may also consult \cite{GelfandKapranovZelevinsky1990,Agostini:2022cgv}.

Let us start by explaining what we mean by a generalized Euler integral. To this end, we first fix the underlying geometric data. Let $\mathcal{A}_1 =\{ a_1^{(1)},\dots, a_{N_1}^{(1)} \}, \dots,\mathcal{A}_m = \{ a_1^{(m)},\dots, a_{N_m}^{(m)} \}$ be finite subsets of the lattice $\ZZ^{n}$. For each $k=1,\dots,m$, we consider the Laurent polynomial in the variables $x_1,\dots, x_n$ on the algebraic torus $(\CC^{\times})^{n}$, depending on the parameters $z_1^{(k)}, \dots, z_{N_k}^{(k)}$, given by
\begin{equation}
f_{k}(x,z^{(k)}) = \sum_{j = 1}^{N_k} z_{j}^{(k)} x_1^{a_{1j}^{(k)}} \cdots x_n^{a_{nj}^{(k)}},
\end{equation}
where $a_{ij}^{(k)}$ denotes the $i$-th coordinate of the vector $a_{j}^{(k)}$. We also fix two parameter vectors $\gamma = (\gamma_1,\dots,\gamma_m)$ and $\nu = (\nu_1,\dots,\nu_n) \in \CC^{n}$. The generalized Euler integral associated with these data is then defined by 
\begin{equation}\label{eq:2.27}
I_{\Gamma}(z; \gamma,\nu) = \int_{\Gamma} \dd x_1 \cdots \dd x_n x_{1}^{\nu_1-1} \cdots x_{n}^{\nu_n-1}  f_{1}(x,z^{(1)})^{-\gamma_1} \cdots f_{m}(x,z^{(m)})^{-\gamma_m}   ,
\end{equation}
where $\Gamma$ is a suitable integration cycle in $(\CC^{\times})^{n}$. Integrals of this form include, as special cases, many of the classical multidimensional Euler integrals, as well as a wide variety of integrals arising naturally in algebraic geometry, singularity theory, and mathematical physics; see, e.g., \cite{passare1998studying,berkesch2014eulermellin,hosono1995mirror}.

We now explain how generalized Euler integrals are connected with GKZ hypergeometric systems. To this end, we must extract, from the data defining the integral, the corresponding data defining a GKZ hypergeometric system. Let $\{e_1,\dots,e_m \}$ denote the canonical basis of the lattice $\ZZ^m$. We then consider the finite subset $\mathcal{A}$ of the lattice $\ZZ^{m+n}$ given by
\begin{equation}
\mathcal{A}=
\left\{
\begin{pmatrix}
e_1\\
a_1^{(1)}
\end{pmatrix},
\dots,
\begin{pmatrix}
e_1\\
a_{N_1}^{(1)}
\end{pmatrix},
\begin{pmatrix}
e_2\\
a_1^{(2)}
\end{pmatrix},
\dots,
\begin{pmatrix}
e_2\\
a_{N_2}^{(2)}
\end{pmatrix},
\dots,
\begin{pmatrix}
e_m\\
a_1^{(m)}
\end{pmatrix},
\dots,
\begin{pmatrix}
e_m\\
a_{N_m}^{(m)}
\end{pmatrix}
\right\}.
\end{equation}
Accordingly, the $\mathcal{A}$-matrix associated with this configuration is
\begin{equation}
A= \left( \begin{array}{c|c|c|c} \begin{array}{cccc} 1& 1 & \cdots&1\\ 0& 0 & \cdots&0\\ \vdots& \vdots & \ddots &\vdots\\ 0& 0 & \cdots&0 \end{array} & \begin{array}{cccc} 0& 0 &\cdots&0\\ 1& 1 & \cdots&1\\ \vdots & \vdots & \ddots  &\vdots\\ 0& 0 & \cdots&0 \end{array} & \begin{array}{c} \cdots \\ \cdots \\ \ddots \\ \cdots \end{array} & \begin{array}{cccc} 0 & 0 &\cdots&0\\ 0& 0 & \cdots&0\\ \vdots& \vdots &\ddots &\vdots\\ 1& 1& \cdots&1 \end{array} \\ \hline A_1&A_2&\cdots&A_m \end{array} \right),
\end{equation} 
where $A_k$ denotes the $\mathcal{A}_k$-matrix for $k=1,\dots,m$. We also consider the parameter vector $\beta = (\gamma_1,\dots,\gamma_m,-\nu_1,\dots,-\nu_n) \in \CC^{m+n}$. As explained in the previous subsection, these data determine a GKZ hypergeometric system and, in the language adopted there, the corresponding GKZ hypergeometric ideal $H_{\mathcal{A}}(\beta)$. A fundamental result, proved in Theorem 2.3 of \cite{MatsubaraHeo2020I}, states that the generalized Euler integral \eqref{eq:2.27} belongs to the solution space $\Sol(H_{\mathcal{A}}(\beta))$. As such, this solution is referred to as the generalized Euler integral representation of the GKZ hypergeometric system defined by $\mathcal{A}$ and $\beta$. 

It is worth pointing out that the special form of the configuration $\mathcal{A}$ has an important consequence in the resonant case. As discussed in the previous subsection, resonance leads to lower-dimensional GKZ hypergeometric systems associated with suitable faces of the polytope $\Delta_{\mathcal{A}}$. In the present setting, these faces often arise by deleting one or more entire blocks from the configuration above. Consequently, the resulting lower-dimensional GKZ system is again associated with a generalized Euler integral of the same form, involving fewer Laurent polynomials. This observation will play an important role in the examples considered later.

\section{GKZ hypergeometric systems and Feynman integrals: a new approach}\label{sec:newApproach}
In this section, we present our proposal for interpreting Feynman integrals within the framework of GKZ hypergeometric systems. Unlike the standard treatments in the literature \cite{delaCruz:2019skx,Klausen2022}, our starting point is not the LP representation \cite{Lee:2013hzt}, but rather the Schwinger parametric representation of Feynman integrals. This perspective will enable us to construct explicit Pfaffian systems which, after a suitable gauge transformation, can be mapped into a \emph{canonical system} of differential equations governing the integral.

\subsection{Schwinger parametric representation}
We begin by recalling the Schwinger parametric representation of a Feynman integral. While this material is well known, a brief review is convenient to fix the notation and the elements involved. We adopt the setup and terminology of \cite{Panzer2015} and \cite{Klausen2022}. Further details and a more complete exposition of the general theory can be found in \cite{Smirnov2012,AbreuBrittoDuhr2022,Weinzierl2022}.

Recall that a Feynman integral associated with a Feynman graph $\Gamma$ having $E$ external edges, $I$ internal edges, and $L$ loops is defined in $D$ space-time dimensions by
\begin{eqnarray}\label{eq:feyn-int}
I_{\Gamma}=\int\prod_{r=1}^L\frac{\dd^D\ell_r}{(2\pi)^D}\prod_{a=1}^I\frac{1}{D_a^{\nu_a}}
\end{eqnarray}
where the $D_a$ are quadratic polynomials in the $\ell$'s and the powers $\nu_a$ could in general be complex, but here we will take them $\nu_a\in\mathbb{Z}$.

The central object in this work is the Schwinger representation of \eqref{eq:feyn-int}, which can be obtained directly using the Schwinger trick (see \cite{Sameshima:2019qpr} for a detailed discussion of the different representations of Feynman integrals). Here, instead, we derive the relevant representation directly from the graph structure. We adopt the setup and terminology of \cite{Panzer2015} and \cite{Klausen2022}. Further details and a more complete exposition of the general theory can be found in \cite{Smirnov2012,AbreuBrittoDuhr2022,Weinzierl2022}.

Let $\Gamma$ be a connected Feynman graph. We will denote its set of vertices by $V$, its set of external edges by $E$, and its set of internal edges by $I$. We cast the relation between these quantities and the number of loops $L$
\begin{equation}
L = I - V + 1.
\end{equation}
This relation says, in particular, that to pass from $\Gamma$ to a subgraph with $V - 1$ edges, the minimum number needed for a connected subgraph containing all vertices of $\Gamma$ to remain connected is to remove exactly $L$ internal edges. A subgraph obtained in this way is what is known as a \emph{spanning tree} of $\Gamma$: a subgraph $T$ that contains all vertices of $\Gamma$ and is connected with exactly $V - 1$ edges. We denote the set of all spanning trees of $\Gamma$ by $T_{\Gamma}$. Similarly, removing $L+1$ internal edges reduces $\Gamma$ to a subgraph with two disjoint connected components, containing no loops and covering all vertices of $\Gamma$. Such a subgraph is called a \emph{spanning $2$-tree} of $\Gamma$. Equivalently, any spanning $2$-tree $T_1 \cup T_2$ can be obtained by removing a single internal edge from a spanning tree $T \in T_{\Gamma}$. The two tree components $T_1$ and $T_2$ induce a partition of the vertex set into two complementary, non-empty subsets $V_1$ and $V_2$ such that $V = V_1 \sqcup V_2$. We denote the set of all spanning $2$-trees of $\Gamma$ by $T_{\Gamma}^{(2)}$.

To represent \eqref{eq:feyn-int}, we assign kinematical and structural data to its edges. To each external edge $e \in E$, we attach an external momentum $k_e$ in a $D$-dimensional Minkowski spacetime, subject to momentum conservation,
\begin{equation}
\sum_{e \in E} k_e = 0.
\end{equation}
To each internal edge $e \in I$, we assign a mass $m_e$ corresponding to the propagated particle, as well as a Schwinger parameter $\alpha_e \in \mathbb{R}_{>0}$. Such parameters allow us to encode the graph's geometry and kinematics into two homogeneous polynomials, $\mathcal{U}(\alpha)$ and $\mathcal{F}(\alpha)$, known as the \emph{first and second Symanzik polynomials}. The first Symanzik polynomial $\mathcal{U}(\alpha)$ is a homogeneous polynomial of degree $L$ defined combinatorially by summing over all spanning trees of $\Gamma$,
\begin{equation}
\mathcal{U}(\alpha) = \sum_{T \in T_{\Gamma}} \prod_{e \notin I} \alpha_e.
\end{equation}
The second Symanzik polynomial $\mathcal{F}(\alpha)$ is homogeneous of degree $L+1$ and incorporates both the external kinematics and the internal masses. It decomposes as
\begin{equation}
\mathcal{F}(\alpha) = \mathcal{F}_{0}(\alpha) + \mathcal{U}(\alpha) \sum_{e \in I} m_e^2 \alpha_e,
\end{equation}
where the kinematical contribution $\mathcal{F}_{0}(\alpha)$ is given by a sum over all spanning $2$-trees,
\begin{equation}
\mathcal{F}_{0}(\alpha) = \sum_{T_1 \cup T_2 \in T_{\Gamma}^{(2)}} \left( \sum_{e \in E_1} k_e \right)^2 \left( \prod_{e \notin I_1 \cup T_2)} \alpha_e \right).
\end{equation}
Here, $E_1$ denotes the set of external edges attached to vertices in $V_1$, so that $\sum_{e \in E_1} k_e$ represents the total momentum flowing into the connected component $T_1$.
With these polynomials at hand, we can state the Schwinger parametric representation of the Feynman integral \eqref{eq:feyn-int}
\begin{eqnarray}\label{eq:I-powers-nu}
I_{\Gamma}= \frac{1}{(4\pi)^{\beta L}}\prod_{e\in I}\frac{1}{\Gamma(\nu_e)}\int_{\alpha_e\geq0}\dd \alpha_e\alpha_e^{\nu_e-1}\frac{e^{-\mathcal{F}(\alpha)/\mathcal{U}(\alpha)}}{\mathcal{U}(\alpha)^{\beta}}     
\end{eqnarray}
where we introduced $\beta=D/2$. 

In its present form, the system is not in the hypergeometric framework we shall consider in this work. Instead, our approach relies on a different formulation, which we discuss in detail in the next paragraph.

\subsection{Derivation of the associated GKZ hypergeometric system}
We now turn to the central stage of our formulation. The transition presented here constitutes the key step in the overall construction developed in this work.  As far as we know, this approach has not appeared before in the literature.

Our starting point is that the parametric representation \eqref{eq:I-powers-nu} can be manipulated further to make the Feynman integral's underlying GKZ hypergeometric structure manifest, recasting $I_{\Gamma}$ as a generalized Euler integral. This amounts to performing a change of variables to eliminate the exponential term from \eqref{eq:I-powers-nu}, where the homogeneity of the Symanzik polynomials $\mathcal{U}(\alpha)$ and $\mathcal{F}(\alpha)$ is precisely what allows one to integrate out the radial factor and eliminate the exponential.

To make this representation explicit, let $N = I$ denote the total number of internal edges of the graph $\Gamma$. We order the Schwinger parameters as $\alpha = (\alpha_1, \dots, \alpha_N) \in \mathbb{R}_{>0}^N$. With this notation, the parametric representation \eqref{eq:I-powers-nu} takes the form
\begin{eqnarray}
I_{\Gamma}= \frac{1}{(4\pi)^{\beta L}}\prod_{i=1}^N\frac{1}{\Gamma(\nu_i)}\int_{\alpha_i\geq0}\dd \alpha_i\alpha_i^{\nu_i-1}\frac{e^{-\mathcal{F}(\alpha)/\mathcal{U}(\alpha)}}{\mathcal{U}(\alpha)^{\beta}}.     
\end{eqnarray}
To perform the integration over the overall radial scale, we introduce a suitable change of variables given by
\begin{eqnarray}\label{eq:new-variables}
(\alpha_1,\dots,\alpha_{N-1},\alpha_N)=(tu_1,\dots,tu_{N-1},t)
\end{eqnarray}
where $u = (u_1, \dots, u_{N-1}) \in \mathbb{R}_{>0}^{N-1}$ parametrizes the affine slice and $t > 0$ acts as the radial scale parameter. The measure of integration changes accordingly under the change of variables as
\begin{eqnarray}
\prod_{i=1}^{N} \dd \alpha_i  \frac{\alpha_i^{\nu_i - 1}}{\Gamma(\nu_i)}\longrightarrow \left(\prod_{i=1}^{N-1} \dd u_i \frac{u_i^{\nu_i - 1}}{\Gamma(\nu_i)} \right)  \dd t t^{\nu-1}
\end{eqnarray}
where we have set $\nu = \sum_{i=1}^{N}\nu_i$. Likewise, the homogeneity of the Symanzik polynomials gives
\begin{eqnarray}
\mathcal{U}(\alpha) &\longrightarrow& t^{L} \widetilde{\mathcal{U}}(u), \\
\mathcal{F}(\alpha) &\longrightarrow& t^{L+1} \widetilde{\mathcal{F}}(u),
\end{eqnarray}
which allow us to define the dehomogenized Symanzik polynomials $\widetilde{\mathcal{U}}(u)$ and $\widetilde{\mathcal{F}}(u)$ on $\mathbb{R}_{>0}^{N-1}$ as
\begin{eqnarray}
\widetilde{\mathcal{U}}(u) = \mathcal{U}(u_1, \dots, u_{N-1}, 1), \\
\widetilde{\mathcal{F}}(u) = \mathcal{F}(u_1, \dots, u_{N-1}, 1),
\end{eqnarray}
this implies that, in the exponential 
\begin{eqnarray}
e^{-\mathcal{F}(\alpha)/\mathcal{U}(\alpha)} \longrightarrow e^{-t \widetilde{\mathcal{F}}(u)/ \widetilde{\mathcal{U}}(u)}. 
\end{eqnarray}
Therefore, the integral in \eqref{eq:I-powers-nu} takes the form
\begin{eqnarray}
I_{\Gamma} = \frac{1}{(4\pi)^{\beta L}}\prod_{i=1}^{N-1} \frac{1}{\Gamma(\nu_i)}\int_{u_i\geq0}\dd u_iu_i^{\nu_i-1}\int_{t\geq0}\dd tt^{\nu-L\beta-1}\frac{e^{-t\widetilde{\mathcal{F}}(u)/\widetilde{\mathcal{U}}(u)}}{\widetilde{\mathcal{U}}(u)^{\beta}}.
\end{eqnarray}
Evaluating the integral over $t$ via the well-known Gamma function integral representation gives the final result
\begin{eqnarray}\label{eq:u-int}
I_{\Gamma} = \frac{1}{(4\pi)^{\beta L}}\prod_{i=1}^{N-1} \frac{\Gamma(\nu-\beta)}{\Gamma(\nu_i)}\int_{u_i\geq0}\dd u_iu_i^{\nu_i-1}\widetilde{\mathcal{U}}(u)^{\nu-(L+1)\beta}\widetilde{\mathcal{F}}(u)^{\beta-\nu}.   
\end{eqnarray}
We remark that the application of the same procedure to the Feynman integral in the LP representation leads to the same result.

Comparing the Euler integral representation \eqref{eq:u-int} with the general definition \eqref{eq:2.27}, we immediately recognize $I_\Gamma$ as a generalized Euler integral in $n = N-1$ variables involving $m = 2$ Laurent polynomials, namely $\widetilde{\mathcal{U}}(u)$ and $\widetilde{\mathcal{F}}(u)$. Consequently, $I_\Gamma$ belongs to the solution space $\operatorname{Sol}(H_{\mathcal{A}}(\beta))$ of a GKZ hypergeometric system determined by this integral representation. The parameter vector $\vec{\beta} = (\gamma_1, \gamma_2, -\nu_1, \dots, -\nu_{N-1}) \in \CC^{N+1}$ is directly read off from the exponents as $\gamma_1 = (L+1)\beta - \nu$ and $\gamma_2 = \nu - L\beta$, while the remaining components coincide with the propagator exponents $\nu_i$. 
The underlying configuration $\mathcal{A}$ is the finite subset of the lattice $\ZZ^{N+1}$ formed by the exponent vectors of the monomials in $\widetilde{\mathcal{U}}(u)$ and $\widetilde{\mathcal{F}}(u)$, augmented by the canonical basis vectors $e_1, e_2 \in \ZZ^2$. Consequently, the associated $\mathcal{A}$-matrix exhibits a characteristic two-block structure given by
\begin{equation}
A =
\left(
\begin{array}{ccc|ccc}
1 & \cdots & 1 & 0 & \cdots & 0 \\
0 & \cdots & 0 & 1 & \cdots & 1 \\ \hline
& A_{1} & & & A_{2} &
\end{array}
\right),
\end{equation}
where $A_{1}$ and $A_{2}$ are the matrices whose columns are the exponent vectors of the monomials appearing in $\widetilde{\mathcal{U}}(u)$ and $\widetilde{\mathcal{F}}(u)$, respectively. It is worth emphasizing that the dependence on the kinematic variables enters the GKZ system exclusively through the coefficients of $\widetilde{\mathcal{F}}(u)$, as the first Symanzik polynomial $\widetilde{\mathcal{U}}(u)$ depends solely on the internal graph topology and possesses purely numerical coefficients.

\section{Examples}\label{sec:examples}
This section is devoted to applying the theory developed above to a series of detailed examples of increasing complexity.

\subsection{Two-mass bubble integral}
Let us start with a warm-up example to show with detail how the approach works. The example is the bubble integral with two different masses. We present the original Feynman integral 
\begin{eqnarray}\label{eq:bub-feyn}
I_{\mathrm{bubble}}=\int\frac{\dd^D\ell}{(2\pi)^D}\frac{1}{(\ell^2-m_1^2)^{\nu_1}((\ell+p)^2-m_2^2)^{\nu_2}}.
\end{eqnarray}
Following the procedure shown in the previous section, the inhomogeneous Symanzik polynomials read
\begin{eqnarray}
\widetilde{\mathcal{F}}(u)&=&m^2_2 +(s+m^2_1+m^2_2)u+m^2_1u^2,\\ 
\widetilde{\mathcal{U}}(u)&=&(1+u).
\end{eqnarray}
here we take the convention $s=-p^2$.

Now, up to a global minus sign due to the conventions, we have the integral \eqref{eq:bub-feyn} turned into
\begin{eqnarray}
I_{\mathrm{bubble}}=\frac{1}{(4\pi)^{\beta}}\frac{\Gamma(\nu_1 + \nu_2 - \beta)}{\Gamma(\nu_1)\Gamma(\nu_1)}\Phi
\end{eqnarray}
where $\Phi$ is the generalized Euler integral to analyze
\begin{eqnarray}
\Phi = \int_{u\geq0}\dd uu^{\nu_1-1}(1+u)^{\nu_1 + \nu_2 - 2\beta}(m^2_2 +(s+m^2_1+m^2_2)u+m^2_1u^2)^{\beta-\nu_1-\nu_2}.    
\end{eqnarray}

For the GKZ system, we promote the parameters of the integral to variables. In this sense, we identify the \emph{non-physical} variables that appear in the $\widetilde{\mathcal{U}}$-polynomial as $w_1=w_2=1$, and the physical ones from the $\widetilde{\mathcal{F}}$-polynomial as $z_1=m^1_2$, $z_2=(s+m^2_1+m^2_2)$, $z_3=m^2_1$. The associated Cayley matrix and vector for the system are
\begin{eqnarray}\label{eq:gkz-bub-full}
A=\begin{pmatrix}
1&1&0&0&0\\
0&0&1&1&1\\
0&1&0&1&2\\
\end{pmatrix},\qquad
\vec{\beta}=
\begin{pmatrix}
\nu_1+\nu_2-2\beta\\
\beta-\nu_1-\nu_2\\
-\nu_1
\end{pmatrix}.    
\end{eqnarray}

The previous matrix has codimension two. If we pursue the solution of the integral using this GKZ system, we encounter the usual issue that the resulting representation involves an excess of variables, which can only be removed at the end (see, e.g., \cite{delaCruz:2019skx}). From our perspective, this also manifests itself in the Pfaffian system through a redundant alphabet, indicating that a reduction associated with a facet of the Newton polytope should be sought. Reductions of GKZ systems for Feynman integrals were recently introduced in \cite{Vanhove:2018mto,Britto:2026qjn} in the context of master integrals. Here, we apply such a reduction instead in the context of $D$-modules.

Working with generalized Euler integrals can show directly the candidates for a reduction of a GKZ system by a facet. Looking at \eqref{eq:gkz-bub-full} we notice that a good candidate for a reduction is to take the reduced $\mathcal{A}$-matrix and constant vector from the block related to the physical variables only, i.e.
\begin{eqnarray}
A=\begin{pmatrix}
1&1&1\\
0&1&2\\
\end{pmatrix},\qquad
\vec{\beta}=
\begin{pmatrix}
\beta-\nu_1-\nu_2\\
-\nu_1
\end{pmatrix}.    
\end{eqnarray}
We proceed with the analysis by calculating the kernel of $A$, which is generated by $l=(1,-2,1)$. From this we have one toric variable, namely
\begin{eqnarray}
x=\frac{z_1z_3}{z^2_2}=\frac{m_1^2m_2^2}{(m_1^2+m_2^2+s)^2},    
\end{eqnarray}
For the proposed solution for the GKZ system \eqref{eq:2.24} we solve $A\vec{\alpha}=\vec{\beta}$ and obtain the following result
\begin{eqnarray}\label{eq:bubble-sol-a}
\Phi=z_1^{\beta-\nu_2}z_2^{-\nu_1}f(x).    
\end{eqnarray}
We plug-in this solution into the toric equation
\begin{eqnarray}\label{eq:toric-bubble}
(\partial_1\partial_3-\partial^2_2)\Phi=0,    
\end{eqnarray}
and that take us to an expected hypergeometric equation
\begin{eqnarray}\label{eq:bubble-hypeq}
x(1-4x)f''(x)+(\beta - \nu_2 +1 -x(6+4\nu_1))f'(x)-\nu_1(\nu_1+1)f(x)=0.
\end{eqnarray}

Looking at \eqref{eq:bubble-hypeq} we can easily determine that its holonomic rank is two, since we can solve $f''(x)$ in terms of $f(x)$ and $f'(x)$. The Frobenius basis for the system is
\begin{eqnarray}
J=\begin{pmatrix}
J_1\\
J_2
\end{pmatrix}   
=\begin{pmatrix}
f\\
xf'
\end{pmatrix},    
\end{eqnarray}
and from \eqref{eq:bubble-hypeq} we obtain its associated first order system
\begin{eqnarray}\label{eq:bubble-GM}
\partial_x\begin{pmatrix}
J_1\\
J_2
\end{pmatrix}
=
\begin{pmatrix}
0&\frac{1}{x}\\
\frac{\nu_1(\nu_1+1)}{1-4x}&\frac{\nu_2-\beta}{x}+\frac{2-4\beta+4\nu_1+4\nu_2}{1-4x}
\end{pmatrix}
\begin{pmatrix}
J_1\\
J_2
\end{pmatrix}.
\end{eqnarray}

The equation \eqref{eq:bubble-GM} is not in canonical form, but it is fuchsian. We see that it can be written completely in $\dd\log$ integration kernels, with the alphabet in toric variables $\{x,1-4x\}$. Now we rewrite \eqref{eq:bubble-GM}
\begin{eqnarray}
\dd J&=&(\dd \Omega) J\notag\\
&=&(\Omega_0\dd \log(x)+\Omega_1\dd \log(1-4x))J
\end{eqnarray}
where
\begin{eqnarray}
\Omega_0=
\begin{pmatrix}
0&1\\
0&\nu_2-\beta
\end{pmatrix},\quad
\Omega_1=
\begin{pmatrix}
0&0\\
\nu_1(\nu_1+1)&2-4\beta+4\nu_1+4\nu_2
\end{pmatrix}.
\end{eqnarray}

In its present state, we can solve the system once we set a suitable boundary condition $J_0$ as an iterated Chen integral
\begin{eqnarray}
J(x)=\mathcal{P}\exp\left(\int_0^x\mathrm{d}\Omega\right)J_0.
\end{eqnarray}
Instead, we will transform the system into its $\epsilon$-form in order to make the iterated-integral structure manifest and organize the solution order by order in $\epsilon$.

Before solving the system explicitly, we highlight a feature of this example. The holonomic rank of the associated GKZ system is strictly smaller than the number of master integrals obtained from integration-by-parts identities for the two-mass bubble integral. In particular, \textsc{LiteRed} \cite{Lee:2013mka} identifies three master integrals, while the holonomic rank is only two. This discrepancy demonstrates that, with a suitable reduction, the dimension of the GKZ solution space need not coincide with the number of IBP master integrals.

We now specialize to the physical case of unit propagator powers, $\nu_1=\nu_2=1$, with $\beta=1-\epsilon$. This specialization will allow us to take this GKZ system to a canonical version, like in the DE approach \cite{Henn:2013pwa}. Remember that our system \eqref{eq:bubble-GM} can be written compactly as $\dd J= PJ$, and for the particular values it takes the form
\begin{eqnarray}\label{eq:bubble-GMValues}
\partial_x\begin{pmatrix}
J_1\\
J_2
\end{pmatrix}
=
\begin{pmatrix}
0&\frac{1}{x}\\
\frac{2}{1-4x}&\frac{6x+\epsilon}{x(1-4x)}
\end{pmatrix}
\begin{pmatrix}
J_1\\
J_2
\end{pmatrix}.
\end{eqnarray}

We now turn to another important observation. Although GKZ provides direct access to the singular locus of the system without requiring the solution of any differential equation, it does not, in general, yield a canonical form. A major obstacle is the appearance of square roots in Feynman integrals \cite{Besier:2018jen}, for which no systematic procedure is currently known to derive the appropriate rationalizing variables directly from the GKZ data. In the present example, we therefore introduce the well-known change of variables that rationalizes the square root $\sqrt{1-4x}$, i.e.
\begin{eqnarray}\label{eq:changeVar}
    x=\frac{y}{(1+y)^2}.
\end{eqnarray}
The $P$ matrix, behaving as a one-form, has its components transforming accordingly under a change of variables $x_i=x_i(y)$, namely
\begin{eqnarray}
 P'_i=\frac{\partial x_j}{\partial y_i} P_j    
\end{eqnarray}
Performing the change of variables in the $P$ matrix from \eqref{eq:bubble-GMValues} we are left with the matrix
\begin{eqnarray}
 P=
\begin{pmatrix}
0 & \frac{1-y}{y (y+1)} \\
 \frac{2}{1-y^2} & \frac{6 y + \epsilon(y+1)^2 }{y-y^3}   
\end{pmatrix}.    
\end{eqnarray}
Now the system is ready to be taken to a canonical form. This action can be performed with the \textsc{Canonica}  package \cite{Meyer:2017joq}. Remembering that the Pfaffian matrix changes under a gauge transformation as follows 
\begin{eqnarray}
 P'= T^{-1} PT-T^{-1}\dd T,
\end{eqnarray}
we have the transformation matrix
\begin{eqnarray}
T=\epsilon^{-1}
\begin{pmatrix}
\frac{y+1}{2 (y-1)} & 0 \\
 \frac{(y+1)(4y + \epsilon(5 + 4y - y^2))}{4 (y-1)^3} & \frac{\epsilon(y+1)^2}{4 (y-1)^2}
\end{pmatrix},    
\end{eqnarray}
and the canonical $\epsilon$-form for the Pfaffian system
\begin{eqnarray}
 P'=\epsilon\begin{pmatrix}
\frac{y-5}{2 (y-1) y} & -\frac{1}{2 y} \\
 \frac{3 (y-5)}{2 (y-1) y} & -\frac{3}{2 y}    
\end{pmatrix}.    
\end{eqnarray}
The $P$ matrix is already written on a single variable for this system and can immediately be expanded in $\dd \log$ integration kernels as follows
\begin{eqnarray}
\dd\Omega = \epsilon\sum_{i=1}^2\Omega_i\dd \log(W_i)   
\end{eqnarray}
with the alphabet $\{W_1,W_2\}=\{y,y-1\}$, and matrices of coefficients
\begin{eqnarray}
\Omega_1=
\begin{pmatrix}
\frac{5}{2} & -\frac{1}{2} \\
\frac{15}{2} & -\frac{3}{2} 
\end{pmatrix}\qquad
\Omega_2=
\begin{pmatrix}
-2 & 0 \\
-6 &  0
\end{pmatrix}.
\end{eqnarray}

Being on an $\epsilon$-form the solution of the system is straightforward as a Chen iterated integral. What is needed is a boundary condition that guarantees that the solution is the original Feynman integral for the bubble, this can be achieved by taking the result of the Feynman integral \eqref{eq:bub-feyn} at the pseudo-threshold, i.e. $x=1/4$ or $y=1$, (see Appendix \ref{app:boundary-conditionA}). Finally, our solution back in the original toric variables has the expected form
\begin{eqnarray}
f(x)=\frac{1}{\sqrt{1-4x}}\left\{\log\left(\frac{1+\sqrt{1-4x}}{1-\sqrt{1-4x}}\right)+\mathcal{O}(\epsilon)\right\},
\end{eqnarray}
and, returning to the physical variables, we remember that
\begin{eqnarray}
    x=\frac{m_1^2m_2^2}{(m_1^2+m_2^2+s)^2}
\end{eqnarray}
and using \eqref{eq:bubble-sol-a}, we obtain
\begin{align}\label{eq:bubble-sol}
    \Phi(m_1^2,m_2^2,s;2-2\epsilon)&=\frac{(m_2^2)^{-\epsilon}\|m_1^2+m_2^2+s\|}{(m_1^2+m_2^2+s) \sqrt{\lambda}}\left\{\log\left(\frac{m_1^2+m_2^2+s+\sqrt{\lambda}}{m_1^2+m_2^2+s-\sqrt{\lambda}}\right)+\mathcal{O}(\epsilon)\right\}\notag\\
    &=\mathrm{sgn}(m_1^2+m_2^2+s)\frac{1}{\sqrt{\lambda}}\left\{\log\left(\frac{m_1^2+m_2^2+s+\sqrt{\lambda}}{m_1^2+m_2^2+s-\sqrt{\lambda}}\right)+\mathcal{O}(\epsilon)\right\}
\end{align}
where $\lambda=\lambda(m_1^2,m_2^2,-s)=(s+m_1^2+m_2^2)^2-4m_1^2m_2^2$ is the \textit{K{\"a}ll{\'e}n} function. After the replacement $s=-p^2$, the result obtained in \eqref{eq:bubble-sol} agrees with the expression in Appendix B of \cite{Abreu:2017mtm}.

\subsection{On-shell massless box integral}
    Another relevant warm-up example is provided by the massless on-shell box integral, with all propagators raised to unit power. Under the coordinate transformation introduced in \eqref{eq:new-variables}, the Symanzik polynomial take the form
    \begin{eqnarray}
        \widetilde{\mathcal{F}}(u_1,u_2,u_3)&=& -su_1u_3-tu_2,\\
        \widetilde{\mathcal{U}}(u_1,u_2,u_3)&=&1+u_1+u_2+u_3,
    \end{eqnarray}
where, for the four point case, we take the four-point Mandelstam variables as
\begin{eqnarray}\label{eq:4p-mand}
s&=&(p_1+p_2)^2,\\
t&=&(p_2+p_3)^2.
\end{eqnarray}
Therefore, $w_1=w_2=w_3=w_4=1$ and $z_1=-s$, $z_2=-t$.

    For this system, however, no nontrivial reduction is possible, since every relevant proper face has a zero-dimensional kernel and therefore gives rise to no nontrivial toric relations. We will thus work with the full Cayley matrix and its corresponding constant vector, namely,

    \begin{eqnarray}
        A=
        \begin{pmatrix}
        1 & 1 & 1 & 1 & 0 & 0 \\
         0 & 0 & 0 & 0 & 1 & 1 \\
         0 & 1 & 0 & 0 & 1 & 0 \\
         0 & 0 & 1 & 0 & 0 & 1 \\
        0 & 0 & 0 & 1 & 1 & 0
        \end{pmatrix},\qquad
        \vec{\beta}=
        \begin{pmatrix}
        4-2\beta\\
        \beta-4\\
        -1\\
        -1\\
        -1
        \end{pmatrix}.
    \end{eqnarray}

    Following a procedure analogous to that used in the bubble example, we chose a generator of the integer kernel of $A$, namely, $l=(-1,1,-1,1,-1,1)$. From this generator, we obtained the corresponding toric variable
    \begin{eqnarray}
        x=\frac{w_2w_4z_2}{w_1w_3z_1}=\frac{t}{s}
    \end{eqnarray}
    and the solution ansatz
    \begin{eqnarray}\label{eq:ansatz-sol}
        \Phi=\frac{w_2^{2-\beta}w_4^{2-\beta}z_1^{\beta-3}}{z_2}f(x),
    \end{eqnarray}
    where $f(x)$ satisfy the following equation
    \begin{multline}\label{eq:system-box}
        (-4+4\beta-\beta^2)f(x) +x(4+4x-4\beta-x\beta+\beta^2)f'(x)+\\
        x^2(6+6x-2\beta-x\beta)f''(x)+x^3(1+x)f^{(3)}(x)=0.
    \end{multline}

    Now, using \textsc{Singular} \cite{GPS01}, we obtained that the holonomic rank of the above system is three. This motivates the following choice of Frobenius basis
    \begin{eqnarray}\label{eq:basis-on-box}
        J=\begin{pmatrix}
       J_1\\
       J_2\\
       J_3
        \end{pmatrix}=\begin{pmatrix}
        f\\
        xf'\\
        x(xf')'\\
        \end{pmatrix},
    \end{eqnarray}
    and from \eqref{eq:system-box}, taking $\beta=2-\epsilon$, we obtained its associated first order system
    \begin{eqnarray}
        \partial_x\begin{pmatrix}
            J_1\\
       J_2\\
       J_3
        \end{pmatrix}= P_x\begin{pmatrix}
            J_1\\
       J_2\\
       J_3
        \end{pmatrix}
    \end{eqnarray}
    where
    \begin{eqnarray}
         P_x=\begin{pmatrix}
            0 & \frac{1}{x} & 0\\
            0 & 0 & \frac{1}{x}\\ 
            0 & 0& \frac{1-x}{x(1+x)}
        \end{pmatrix}+\epsilon\begin{pmatrix}
            0& 0&0\\
            0& 0& 0\\ 
            0&\frac{2}{x(1+x)}&-\frac{x+2}{x(1+x)}
        \end{pmatrix}+\epsilon^2\begin{pmatrix}
            0 & 0& 0\\
            0& 0 &0 \\
            \frac{1}{x(1+x)} & -\frac{1}{x(1+x)} &0
        \end{pmatrix}.
    \end{eqnarray}
    The differential system obtained above is not yet in a canonical form. However, a further application of the \textsc{Canonica} package,  allow us to find a gauge transformation to a canonical basis as $J = T \:G$, where
    \begin{eqnarray}
        T(x;\epsilon)=\begin{pmatrix}
            1 & 0 & 0\\
            -\frac{\epsilon}{2} & -\frac{\epsilon}{2} & \frac{\epsilon(x-1)}{2(1+x)}\\
            0 & \frac{\epsilon^2}{1+x} &\frac{\epsilon(\epsilon+x)}{(1+x)^2}
        \end{pmatrix}
    \end{eqnarray}
    and $\partial_xG=d\Omega\:G$, where
    \begin{eqnarray}
        d\Omega=\epsilon\sum_{i=1}^{2}\Omega_i\:\dd \log(W_i),
    \end{eqnarray}
    with $\{W_1,W_2\}=\{x,1+x\}$ and
    \begin{eqnarray}
        \Omega_1=\begin{pmatrix}
            -\frac{1}{2}&-\frac{1}{2}&-\frac{1}{2}\\
            \frac{1}{2}&-\frac{1}{2} &-\frac{1}{2}\\
            0 & -1 & -1
        \end{pmatrix},\quad \Omega_2=\begin{pmatrix}
            0 & 0 &1\\
            0 & 0 & 0\\
            0 & 0& 1
        \end{pmatrix}.
    \end{eqnarray}
    We chose the boundary conditions for our system to be $f(x=-1)$ and $f'(x=-1)$. Since $x=-1$ corresponds to the kinematic condition $s+t=0$, determining these values amounts to calculating the Feynman integral and its first derivative at the pseudo-threshold. At this point (see Appendix \ref{app:boundary-conditionB}), we have
    \begin{eqnarray}
        f(x=-1)=\frac{4}{\epsilon^2}+\mathcal{O}\left(\epsilon^{-1}\right) \quad \text{and}\quad f'(x=-1)=\frac{2}{\epsilon}+\mathcal{O}\left(\epsilon^0\right).
    \end{eqnarray}
    Hence, under the gauge transformation, we have that the boundary condition in the new basis is $G(x=-1)=(4,0,0)^T/\epsilon^2+\mathcal{O}(\epsilon^{-1})$. Finally, transforming back to the original basis and restoring the physical variables, we get
    \begin{multline}\label{eq:sol-mlessbox}
        \Phi(s,t;4-2\epsilon)=\frac{(-s)^{-\epsilon}}{st}\Bigg\{\frac{4}{\epsilon^2}-\frac{2}{\epsilon}\log\left(\frac{t}{s}\right)+\frac{4}{\epsilon}+2\log\left(\frac{t}{s}\right)-\frac{5\pi^2}{3}\\+\epsilon\Bigg[2\operatorname{Li}_3\left(-\frac{s}{t}\right)+2\log\left(\frac{t}{s}\right)\operatorname{Li}_2\left(-\frac{s}{t}\right)+\log^3\left(\frac{t}{s}\right)\\
        -\left(\log^2\left(\frac{t}{s}\right)+\pi^2\right)\log\left(\frac{s+t}{s}\right)\\+\left(\frac{4\pi^2}{3}-2\right)\log\left(\frac{t}{s}\right)-4+\frac{5\pi^2}{3}-10\zeta(3)\Bigg]+\mathcal{O}\left(\epsilon^{2}\right)\Bigg\}.
    \end{multline}
    where $\operatorname{Li}_n$ is the weight-n classical polylogarithm and $\zeta(n)$ is the Riemann zeta constant.
    
    The result obtained above agrees with those reported in \cite{Henn:2014qga,Calisto:2023vmm}, up to an overall constant kinematics-independent normalization factor. In particular, if $F_4$ denotes the convention used in the mentioned references, then
    \begin{eqnarray}
        F_4=e^{\gamma_E\epsilon}\Gamma(2+\epsilon)\:\Phi
    \end{eqnarray}
    where $\gamma_E$ is the Euler--Mascheroni constant.
    
\subsection{Off-shell massless triangle and box integrals}

We now move to examples with co$(\mathcal{A})=2$. Starting with the off-shell massless triangle integral with all the powers of the propagators equal to one. The Symanzik polynomials take the form
\begin{eqnarray}
\widetilde{\mathcal{F}}(u_1,u_2)&=& su_1+uu_2+tu_1u_2,\\
\widetilde{\mathcal{U}}(u_1,u_2)&=&1+u_1+u_2,
\end{eqnarray}
and we used the following convention for the Mandelstam variables in the three-point case
\begin{eqnarray}
s&=&-(p_1+p_2)^2=-p_3^2,\\
u&=&-(p_1+p_3)^2=-p_2^2,\\
t&=&-(p_2+p_3)^2=-p_1^2.
\end{eqnarray}
Therefore, $w_1=w_2=w_3=1$ and $z_1=s, z_2=u,z_3=t$.

For this system there is no need for a reduction, due to the size of the polynomial, and the respective Cayley matrix and constant vector are
\begin{eqnarray}
A=
\begin{pmatrix}
1 & 1 & 1 & 0 & 0 & 0 \\
 0 & 0 & 0 & 1 & 1 & 1 \\
 0 & 1 & 0 & 1 & 0 & 1 \\
 0 & 0 & 1 & 0 & 1 & 1
\end{pmatrix},\qquad
\vec{\beta}=
\begin{pmatrix}
3-2\beta\\
\beta-3\\
-1\\
-1
\end{pmatrix}
\end{eqnarray}
Following the procedure, from the kernel of $A$, we arrive at the toric variables
\begin{eqnarray}
x&=&\frac{w_1z_3}{w_3z_1}=\frac{t}{s},\\
y&=&\frac{w_2z_2}{w_3z_1}=\frac{u}{s},
\end{eqnarray}
and the solution ansatz
\begin{eqnarray}
\Phi=\frac{w_1^{2-\beta } w_2^{2-\beta } z_1^{\beta -3}}{w_3}f(x,y).
\end{eqnarray}
The hypergeometric system for $f(x,y)$, from the toric equations, consist on the following system of equations 
\begin{eqnarray}\label{eq:hyp-tri}
(\beta -3)f+(\beta -5) y \partial_yf+(\beta  (x-1)-5 x+3)\partial_xf-y \left(y \partial^2_yf+2 x\partial^2_{xy}f\right)-(x-1) x \partial^2_xf&=&0,\notag\\
(\beta -3)f+(\beta -5) x\partial_xf+(-\beta +(\beta -5) y+3) \partial_yf-2 x y\partial^2_{xy}f-(y-1) y \partial^2_yf-x^2 \partial^2_xf&=&0.\notag\\
\end{eqnarray}
The holonomic rank of the system was computed using \textsc{Singular}, yielding a rank of four. This allow us to chose the Frobenius basis for two variables in the following way
\begin{eqnarray}\label{eq:basis-tri}
J=\begin{pmatrix}
f\\
x\partial_xf\\
y\partial_yf\\
xy\partial^2_{xy}f
\end{pmatrix}.
\end{eqnarray}
Using \eqref{eq:hyp-tri} we arrive at the respective system for the basis with the form
\begin{eqnarray}
\partial_x J &=&  P_xJ,\\
\partial_y J &=&  P_yJ,
\end{eqnarray}
where the $4\times4$ $ P_i$ matrices for the first order system have large expressions, but with the following structure
\begin{eqnarray}\label{eq:GM-tri}
 P_i= P_i^{(0)}(x,y)+\epsilon P_i^{(1)}(x,y),
\end{eqnarray}
and with the toric alphabet $\left\{x,y,x+y-1,(1-x-y)^2 - 4xy\right\}$. The explicit expressions for the matrices can be seen in Appendix \ref{app:explicitGM}.

From the previous alphabet we can immediately recognize the \textit{K{\"a}ll{\'e}n} function for the massless triangle integral, $\lambda(x,y) = (1-x-y)^2 - 4xy$. In order to rationalize $\sqrt{\lambda(x,y)}$ we introduce the change of variables
\begin{eqnarray}
x=\frac{u}{1+v},\qquad y=\frac{v(1-u)}{1+v}.
\end{eqnarray}

These new variables allow us to find a gauge transformation to a canonical basis as $J=TG$ (see also Appendix \ref{app:explicitGM}), and here we write the final result for the Pfaffian system in $\epsilon$-form $dG=\dd\Omega G$, with
\begin{eqnarray}
\dd\Omega = \epsilon\sum_{i=1}^5\Omega_i\dd \log(W_i)   
\end{eqnarray}
now the alphabet in the new variables and in the respective order is $\{W_1,W_2,W_3,W_4,W_5\}=\{u,v,u-1,v+1,uv+u-1\}$. For the coefficient matrices we have

\begin{eqnarray}
&&\Omega_1=
\begin{pmatrix}
-2 & 0 & -1 & 0 \\
 -2 & 0 & -2 & 1 \\
 2 & 0 & 1 & 0 \\
 2 & 0 & 2 & -1
\end{pmatrix},\qquad
\Omega_2=
\begin{pmatrix}
-2 & -1 & -2 & 0 \\
 -2 & -3 & -2 & -2 \\
 2 & 2 & 2 & 1 \\
 2 & 2 & 2 & 1
\end{pmatrix},\qquad
\Omega_3=
\begin{pmatrix}
0 & 1 & 2 & 0 \\
 0 & -1 & 2 & -2 \\
 0 & 0 & -2 & 1 \\
 0 & 0 & -2 & 1
\end{pmatrix},\notag
\\
&&\Omega_4=
\begin{pmatrix}
 0 & 1 & 1 & 0 \\
 0 & 3 & 2 & 1 \\
 0 & -2 & -1 & -1 \\
 0 & -2 & -2 & 0
\end{pmatrix},\qquad
\Omega_5=
\begin{pmatrix}
2 & 0 & 0 & 0 \\
 2 & 0 & 0 & 0 \\
 -2 & 0 & 0 & 0 \\
 -2 & 0 & 0 & 0
\end{pmatrix}\notag.\\
\end{eqnarray}

The boundary condition we take for our system is $f(x=0,y=0)=\tilde{f}(u=0,v=0)$, which falls into the well-known single scale triangle integral. All the information we need from it is, that at leading order, it give us a behavior like
\begin{eqnarray}
\tilde{f}(0,0)=\frac{1}{\epsilon^2}.
\end{eqnarray}
Under the gauge transformation, we have the boundary condition on the new basis as $(-1, -2, 2, 2)^\mathrm{T}/\epsilon^2$.
Now we integrate the system using the \textsc{PolyLogTools} package \cite{Duhr:2019tlz}, and write the answer compactly in terms of Goncharov's polylogarithms \cite{Goncharov:1998kja,Goncharov:2001iea}, i.e.
\begin{eqnarray}
G(a_1,\ldots,a_n;z)=\int_0^z \frac{dt}{t-a_1}\,G(a_2,\ldots,a_n;t),\qquad G(z)=1.
\end{eqnarray}

Back to our problem, we have the solution
\begin{eqnarray}
\tilde{f}(u,v)=\frac{1}{\sqrt{\lambda(u,v)}}\left(\frac{1}{\epsilon^2} + \frac{1}{\epsilon}\tilde{f}_{-1}(u,v) + \tilde{f}_0(u,v) \right)
\end{eqnarray}
where
\begin{eqnarray}\label{eq:mless-tri}
\tilde{f}_{-1}(u,v)&=& G(1;u)-G\left(\frac{1}{1+v};u\right),\notag\\
\tilde{f}_{0}(u,v)&=& 2 \Bigg(G(0,v) \left(G(1,u)-G\left(\frac{1}{v+1},u\right)\right)-2G\left(1,\frac{1}{v+1},u\right)-G\left(\frac{1}{v+1},0,u\right)\notag\\
&&-G\left(\frac{1}{v+1},1,u\right)+2G\left(\frac{1}{v+1},\frac{1}{v+1},u\right)+G(1,0,u)+G(1,1,u)\Bigg).\notag\\
\end{eqnarray}
Putting everything together, we finally have the solution for the integral 
\begin{eqnarray}
\Phi=s^{-1-\epsilon}\tilde{f}(u,v).
\end{eqnarray}

The massless off-shell box case is relevant here because it provides us with another good example of a reduction of a GKZ system and allows us to see its relation to the previous triangle integral as a good test. Due to the similarities, we will only highlight the important steps. The box integral in question has the structure
\begin{eqnarray}\label{eq:euler-box}
\Phi=\int_{u_i\geq0}\dd u_1\dd u_2\dd u_3\;\widetilde{\mathcal{U}}(u_1,u_2,u_3)^{4-2\beta}\widetilde{\mathcal{F}}(u_1,u_2,u_3)^{\beta-4}
\end{eqnarray}
We cast the Symanzik polynomials for the box integral in the $u_i$ variables
\begin{eqnarray}
\mathcal{F}(u_1,u_2,u_3)&=& p_4^2u_1+tu_2+p_1^2u_1u_2+p_3^2u_3+su_1u_3+p_2^2u_2u_3,\\
\mathcal{U}(u_1,u_2,u_3)&=&1+u_1+u_2+u_3,
\end{eqnarray}
where the Mandelstam variables $s$ and $t$ were defined in \eqref{eq:4p-mand}.

The original Cayley system for the integral \eqref{eq:euler-box} consists of a $10\times5$ matrix. We will not present it here and go directly to the system reduced by a facet. The respective $\mathcal{A}$-matrix for the reduced system and constant vector are
\begin{eqnarray}\label{eq:red-box}
A=
\begin{pmatrix}
1 & 1 & 1 & 1 & 1 & 1 \\
 1 & 0 & 1 & 0 & 1 & 0 \\
 0 & 1 & 1 & 0 & 0 & 1 \\
 0 & 0 & 0 & 1 & 1 & 1
\end{pmatrix},\qquad
\vec{\beta}=
\begin{pmatrix}
\beta-4\\
-1\\
-1\\
-1
\end{pmatrix}
\end{eqnarray}
For this system we have only six physical parameters denoted as $z_1=p_4^2,z_2=t,z_3=p_1^2,z_4=p_3^2,z_5=s,z_6=p_2^2$.

The kernel of the reduced $\mathcal{A}$-matrix in \eqref{eq:red-box} gives us a noticeable reduction to only two toric variables
\begin{eqnarray}
X&=&\frac{z_3z_4}{z_2z_5}=\frac{p_1^2p_3^2}{st},\\
Y&=&\frac{z_1z_6}{z_2z_5}=\frac{p_4^2p_2^2}{st},
\end{eqnarray}
and the proposed solution to \eqref{eq:euler-box} is
\begin{eqnarray}
\Phi=\frac{(z_1z_4)^{\beta-2} z_5^{1-\beta }}{z_2}f(X,Y).
\end{eqnarray}

By following the procedure we will arrive at a system with holonomic rank four. From this, the basis will be the same as that for the triangle \eqref{eq:basis-tri}, but with the renaming of the variables as $(x,y)\longrightarrow (X,Y)$.  The first order system will present a slight difference in the sign of the Pfaffian matrix from \eqref{eq:GM-tri}; it now takes the form 
\begin{eqnarray}
 P_i= P_i^{(0)}(X,Y)-\epsilon P_i^{(1)}(X,Y),
\end{eqnarray}
with the same alphabet for $(X,Y)$. Just to differentiate from the triangle case, to rationalize the alphabet, we will use the same change of variables with capital letters 
\begin{eqnarray}
X=\frac{U}{1+V},\qquad Y=\frac{V(1-U)}{1+V}.
\end{eqnarray}

The result of repeating the steps for the triangle integral takes us to the solution of the box integral. With the boundary condition being the on-shell massless box of \eqref{eq:sol-mlessbox}, the solution for $f$ reads
\begin{eqnarray}
\tilde{f}(U,V)=\frac{1}{\sqrt{\lambda(U,V)}}\left(-\frac{2}{\epsilon^2}  -4\frac{1}{\epsilon}\tilde{f}_{-1}(U,V) + 4\tilde{f}_0(U,V) \right)
\end{eqnarray}
where the terms of the expansion are the ones for the triangle in \eqref{eq:mless-tri}. Finally, the result for the integral is
\begin{eqnarray}
\Phi=\frac{(p_4^2p_3^2t)^{\epsilon}}{st}\tilde{f}(U,V),
\end{eqnarray}
and with this we recover the functional behavior for the massless off-shell triangle and box integral that was found originally in \cite{Usyukina:1992jd}. Again, an important remark for the box example is that the reduction in the size of basis to solve the system here is four, while in the DE approach the basis consists of eleven master integrals. 

\subsection{Sunset integral with one massless propagator}

As a final example, for another system of codimension two, we will present the $\epsilon$-form for the sunset integral with one massless propagator. The explicit and large expressions will be presented in Appendix \ref{app:explicitGM}. The integral has the following Symanzik polynomials
\begin{eqnarray}
\widetilde{\mathcal{F}}(u_1,u_2)&=& m_3^2 u_1+m_3^2u_2+(m_2^2+m_3^2+s)u_1 u_2 +m_2^2u_2^2+m_2^2 u_1u_2^2,\\
\widetilde{\mathcal{U}}(u_1,u_2)&=&u_1+u_2+u_1u_2,
\end{eqnarray}

Moving forward, the initial system has holonomic rank three, so the Frobenius basis is $J=(f,x\partial_xf,y\partial_yf)^{\mathrm{T}}$, where the toric variables are
\begin{eqnarray}
x&=&\frac{m_2^2 m_3^2}{(m_2^2+m_3^2+s)^2},\\
y&=&\frac{m_2^2}{(m_2^2+m_3^2+s)}.
\end{eqnarray}

The first order system has the expansion in $\epsilon$
\begin{eqnarray}
\partial_iJ=\left( P(x,y)^{(0)}_i+\epsilon P(x,y)^{(1)}_i+\epsilon^2 P(x,y)^{(2)}_i\right)J,
\end{eqnarray}
with the toric alphabet $\{x,y,4x-1,x-y+y^2\}$.

In order to find the canonical form, we introduce the change of variables
\begin{eqnarray}
x=\frac{uv}{v+v},\qquad y=\frac{(u+v)^2-u^3v}{u(2-u)(u+v)^2}.
\end{eqnarray}

Then, on the new variables, the gauge transformation is also quadratic in $\epsilon$
\begin{eqnarray}
T(u,v)=T(u,v)^{(0)}+\epsilon T(u,v)^{(1)}+\epsilon^2T(u,v)^{(2)},
\end{eqnarray}
and it leads to the first order system in $\epsilon$-form 
\begin{eqnarray}
\dd \Omega = \epsilon\sum_{i=1}^7\Omega_i\:\dd \log(W_i)\; G,
\end{eqnarray}
with respective alphabet in order is $\{W_1,W_2,W_3,W_4,W_5,W_6,W_7\}=\{u,v,u-2,u+v,u^2-u-v,uv-u-v,u^3v-u^2-2uv- v^2\}$. For the coefficient matrices we have
\begin{eqnarray}
&&\Omega_1=
\begin{pmatrix}
3 & -1 & 8 \\
 0 & -1 & 0 \\
 0 & -\frac{1}{4} & 1
\end{pmatrix},\qquad
\Omega_2=
\begin{pmatrix}
0 & 1 & -8 \\
 0 & 0 & 0 \\
 0 & -\frac{1}{4} & 2
\end{pmatrix},\qquad
\Omega_3=
\begin{pmatrix}
1 & 0 & 0 \\
 0 & -1 & 0 \\
 0 & -\frac{1}{4} & 1
\end{pmatrix},\notag
\\
&&\Omega_4=
\begin{pmatrix}
0 & 0 & 0 \\
 0 & -2 & 0 \\
 0 & -\frac{1}{4} & 0
\end{pmatrix},\qquad
\Omega_5=
\begin{pmatrix}
-2 & 0 & -8 \\
 0 & 0 & 0 \\
 0 & 0 & 0
\end{pmatrix}\qquad,
\Omega_6=
\begin{pmatrix}
0 & 0 & 8 \\
 0 & 0 & 0 \\
 0 & 0 & -2
\end{pmatrix},\notag\\
&&\Omega_7=
\begin{pmatrix}
0 & 0 & 0 \\
 0 & 1 & 0 \\
 0 & \frac{1}{4} & 0
\end{pmatrix},\notag\\
\end{eqnarray}

For this example, we also computed the number of master integrals using \textsc{LiteRed} \cite{Lee:2013mka} and found it to be four. This provides a further example in which the holonomic rank is strictly smaller than the number of IBP master integrals, reinforcing the central observation of this work. With this, we conclude our collection of examples demonstrating the application of the proposed approach.

\section{Conclusions}\label{sec:concl}

In this work, we have used the technology of $D$-modules to establish a direct connection between GKZ systems for Feynman integrals and canonical differential equations. While the LP representation has so far provided the natural setting for connecting Feynman integrals with GKZ hypergeometric systems, we have shown here that the generalized Euler integral arising from the Schwinger representation provides an alternative and particularly useful starting point. In this formulation, reductions associated with facets of the Newton polytope can be identified more directly, leading to a reduction in the number of variables and making the GKZ framework especially well suited to the study of Feynman integrals. We have exploited this feature in the construction of the system of differential equations.

Starting from the GKZ hypergeometric system, the associated Pfaffian system was obtained by constructing a Frobenius basis. This procedure also provides a direct identification of the singular locus of the Feynman integral. To bring the resulting system into canonical $\epsilon$-form, an additional change to rationalizing variables is required. An important feature of the construction is that the dimension of the canonical basis is determined by the holonomic rank of the GKZ hypergeometric system. In the examples considered here, this dimension is smaller than or equal to the number of master integrals obtained from the corresponding IBP reduction. This provides an indication that the $D$-module/GKZ framework can offer a more economical description of the space of functions associated with a Feynman integral, while providing a direct route from its integral representation to its canonical differential equation.

A work in progress concerns the construction of different bases from the Frobenius ones, following a new approach and avoiding the need of a rationalizing change of variables. Other interesting avenues include studying the vector space generated by the generalized Euler integrals arising from the Schwinger representation, using the constructions of twisted De Rham cohomology \cite{Agostini:2022cgv} and their possible motivic extensions \cite{Marcolli:2009zy}. Concurrently, efforts are focused on the calculation of one-loop integrals with higher multiplicity as well as multi-loop cases.

The polylogarithmic examples analyzed in this work provide a first setting in which to explore the connection between GKZ systems and canonical differential equations from the generalized Euler representation. A natural and immediate direction for future work is to extend this analysis to elliptic Feynman integrals. In this case, the corresponding $D$-module structure and the associated GKZ systems are expected to provide a natural framework for investigating how the present construction extends beyond the polylogarithmic regime, in particular in relation to the appearance of elliptic periods and the differential equations governing them \cite{Adams:2018yfj,Frellesvig:2021hkr,Dlapa:2022wdu,Gorges:2023zgv}.

\section*{Acknowledgments}

We are thankful to Leonardo de la Cruz, William Torres Bobadilla and Pierre Vanhove for valuable comments and important questions in the final version of the draft.




\appendix

\section{Proof of the Pfaffian reduction lemma}\label{app:A}
In this appendix, we provide the explicit proof of the claim made at the end of Section \ref{sec:2.2} regarding the finite-dimensional vector space structure of the quotient ring $\tilde{R}_n / I_{\mathcal{A}}(\beta)$ after performing the separation of variables. Throughout this appendix, we strictly adopt the notation and conventions established in Subsection \ref{sec:2.2}. The concrete statement of the result is as follows.

\begin{lemma}
The quotient $\tilde{R}_n / I_{\mathcal{A}}(\beta)$ is a finite-dimensional vector space over the field $\CC(x_1,\dots,x_n)$, whose dimension coincides with the holonomic rank of the original GKZ hypergeometric ideal $H_{\mathcal{A}}(\beta)$.
\end{lemma}

\begin{proof}
Let $\tau \colon R_N \to R_N$ be the $\CC(z_1,\dots,z_N)$-algebra automorphism defined by $\tau(\partial_j) = \partial_j + \alpha_j z_j^{-1}$ for $j = 1, \dots, N$. Since \(A\alpha = \beta\), under such automorphism, the first-order differential operators \(Z_{i}\) are mapped to the first-order differential operators $\tau(Z_i) = \sum_{j=1}^N a_{ij} z_j \partial_j$. The image ideal $\tau(H_{\mathcal{A}}(\beta))$ is thus generated by these first-order differential operators together with the transformed box operators $\tau(\Box_l)$.

To analyze the quotient by these first-order differential operators, consider the vector space over $\CC(z_1,\dots,z_N)$ of logarithmic derivations. The standard basis is given by $\{z_1\partial_1, \dots, z_N\partial_N\}$. On the other hand, by the chain rule and the definition of the toric variables, the operators $\theta_j$ can be written as $\theta_{j}=\sum _{k=1}^{N}l_{jk}z_{k}\partial_{k}$. Thus, the coefficients of the operators $\tau(Z_1),\dots,\tau(Z_m)$ are given by the rows of $A$, while those of $\theta_1\dots,\theta_n$ are given by the basis vectors $l_1,\dots,l_n$ of $L$. Since these $N$ vectors of coefficients form a basis of $\CC^N$, the corresponding operators $\tau(Z_1), \dots, \tau(Z_m), \theta_1, \dots, \theta_n$ form an alternative basis of the space of logarithmic derivations. Moreover, because they are linear combinations of the commuting operators $z_j \partial_j$, they commute with one another, and hence their ordered monomials form a basis for $R_{N}$ over $\CC(z_1,\dots,z_N)$. From this, modulo the left ideal generated by $\tau(Z_1), \dots, \tau(Z_m)$, every differential operator has a unique representative which is a linear combination of monomials in $\theta_1, \dots, \theta_n$. Therefore, after quotienting $R_{N}$ by the first-order generators of $\tau(H_{\mathcal{A}}(\beta))$, the quotient is identified with the algebra generated by $\theta_1, \dots, \theta_n$, with coefficients in $\CC(z_1,\dots,z_N)$.

We now have to impose the remaining relations coming from $\tau(H_{\mathcal{A}}(\beta))$, namely those defined by the transformed box operators $\tau(\Box_l)$ for $l \in L$. Because each original box operator $\Box _{l}$ is homogeneous of degree zero with respect to the algebraic torus action defined by the matrix $A$, its image under the map $\tau$ remains invariant under this action. Consequently, the transform $\tau(\Box_l)$ preserves the quotient by the first-order operators $\tau(Z_i)$. Under the identification established above, this invariance ensures that the coefficients of the induced operator, which a priori lie in $\CC(z_1,\dots,z_N)$, belong to the subfield of invariant rational functions $\CC(x_1,\dots,x_n)$. Thus, the induced operator is expressed entirely in terms of the toric variables $x_1, \dots, x_n$ and the operators $\theta_1, \dots, \theta_n$. By the construction of the ideal $I_{\mathcal{A}}(\beta)$, these induced operators are precisely its generators. Consequently, imposing the relations coming from the transformed box operators on the quotient by the operators \(\tau(Z_i)\) yields the canonical isomorphism of left modules 
$$
R_N / \tau(H_{\mathcal{A}}(\beta)) \cong \CC(z_1, \dots, z_N) \otimes_{\CC(x_1, \dots, x_n)} ( \tilde{R}_n / I_{\mathcal{A}}(\beta) ).
$$
The dimension of the left-hand side over $\CC(z_1, \dots, z_N)$ is therefore equal to the dimension of $\tilde{R}_n / I_{\mathcal{A}}(\beta)$ over $\CC(x_1, \dots, x_n)$. Since $\tau$ is an automorphism, the former is precisely the holonomic rank of $H_{\mathcal{A}}(\beta)$, establishing the desired equality of dimensions.
\end{proof}

\section{Boundary conditions at the pseudo-threshold }\label{app:bund-cond}

\subsection{Two-mass bubble}
\label{app:boundary-conditionA}

The boundary singular regular point $y=1$ corresponds, under the change of variables \eqref{eq:changeVar}, to $x=1/4$. Since
\begin{eqnarray}
    x=\frac{z_1z_3}{z_2^2},
\end{eqnarray}
we choose the branch
\begin{eqnarray}
    z_2=2\sqrt{z_1z_3}.
    \label{eq:pseudothreshold-condition}
\end{eqnarray}
For \(\nu_1=\nu_2=1\) and \(\beta=1-\epsilon\), the generalized Euler
representation of the GKZ solution is
\begin{eqnarray}
    \Phi(z_1,z_2,z_3)
    =\int_0^\infty
    \,du\frac{(1+u)^{2\epsilon}}
    {\left(z_1+z_2u+z_3u^2\right)^{1+\epsilon}},
\end{eqnarray}
while
\begin{eqnarray}
    f(x)=z_1^\epsilon z_2\,\Phi(z_1,z_2,z_3).
\end{eqnarray}
On the locus \eqref{eq:pseudothreshold-condition}, the quadratic polynomial
degenerates into a perfect square,
\begin{eqnarray}
    z_1+z_2u+z_3u^2
    =
    \left(\sqrt{z_1}+\sqrt{z_3}\,u\right)^2.
\end{eqnarray}
Consequently,
\begin{align}
    \left.\Phi\right|_{x=1/4}
    &=
    \int_0^\infty
    \,du\frac{(1+u)^{2\epsilon}}
    {\left(\sqrt{z_1}+\sqrt{z_3}\right)^{2+2\epsilon}}\notag\\
    &=
    \frac{1}{1+2\epsilon}
    \frac{
        z_1^{-(1+2\epsilon)/2}
        -
        z_3^{-(1+2\epsilon)/2}
    }{
        \sqrt{z_3}-\sqrt{z_1}
    },
    \label{eq:Phi-boundary}
\end{align}
where the integral has been evaluated by the substitution $t=(\sqrt{z_1}+\sqrt{z_3}u)/(1+u)$. Hence, at \(x=1/4\), the transformation to the canonical basis $G(y,\epsilon)$ gives
\begin{equation}
    G_1(x=1/4)
    =\left.-2\epsilon\sqrt{1-4x}\;f\left(x\right)\right|_{x=1/4}=0,
\end{equation}
and,
\begin{align}
    G_2(x=1/4)&=\left.-2\left[4x+\epsilon\left(2+3\sqrt{1-4x}\right)\right]f(x)+4\epsilon(1-4x)xf_x(x)\right|_{x=1/4}\notag\\
    &=-2(1+2\epsilon)f(1/4)\notag\\
    &=-2(1+2\epsilon)z_1^\epsilon z_2\,\left.\Phi\right|_{x=1/4}.
\end{align}
Using \eqref{eq:Phi-boundary} and \(z_2=2\sqrt{z_1z_3}\), one obtains
\begin{equation}
    G(y=1;\epsilon)
    =
    \begin{pmatrix}
        0\\
        \xi(\epsilon)
    \end{pmatrix},\quad \text{where}
    \quad
    \xi(\epsilon)
    =
    -4\left(
    \frac{1-r^{1+2\epsilon}}{1-r}\right), \quad \text{and}
    \quad
    r=\sqrt{\frac{z_1}{z_3}}.
    \label{eq:boundary-xi}
\end{equation}

\subsection{On-shell massless box}
\label{app:boundary-conditionB}
From \eqref{eq:ansatz-sol} and taking $\beta=2-\epsilon$, we have that
    \begin{eqnarray}
            f(x)=\frac{z_1^{1+\epsilon}z_2}{(w_2w_4)^\epsilon}\Phi(z_1,z_2).
    \end{eqnarray}
At the pseudothreshold $x=-1$, which corresponds to $t=-s$ in terms of the physical kinematic variables, the boundary value becomes
\begin{eqnarray}\label{eq:evalf}
    f(x=-1)=\left.\frac{z_1^{1+\epsilon}z_2}{(w_2w_4)^\epsilon}\Phi(z_1,z_2)\right|_{w_2=w_4=1,z_1=-z2=-s}.
\end{eqnarray}
In order to evaluate the above expression, first, we can calculate $\Phi(z_1=-z_2=-s)\equiv\left.\Phi\right|_{x=-1}$. Notice that,
\begin{align}
    \Phi&=\int_{u_i\geq0}\dd u_1\dd u_2 \dd u_3\;(1+u_1+u_2+u_3)^{2\epsilon}(-su_1u_3-tu_2)^{-2-\epsilon}\\
    &=\int_{0}^{1}\dd\alpha_1\dd\alpha_2\dd\alpha_3\dd\alpha_4\:\delta\left(1-\sum_{i=1}^{4}\alpha_i\right)(-s\alpha_1\alpha_3-t\alpha_2\alpha_4)^{-2-\epsilon}
\end{align}
under the change of variables given by
\begin{eqnarray}
     \alpha_4=\frac{1}{1+u_1+u_2+u_3}\quad \text{and} \quad \alpha_i=u_i\alpha_4,\quad \text{for}\quad i=1,2,3.
\end{eqnarray}
Now, taking $t=-s$, we have that
\begin{eqnarray}
    \left.\Phi\right|_{x=-1}=\int_{0}^{1}\dd\alpha_1\dd\alpha_2\dd\alpha_3\dd\alpha_4\:\delta\left(1-\sum_{i=1}^{4}\alpha_i\right)(-s\alpha_1\alpha_3+s\alpha_2\alpha_4-i0)^{-2-\epsilon}.
\end{eqnarray}
Assuming $s>0$, we introduce a further change of variables,
\begin{eqnarray}\label{eq:newcoord}
    \alpha_1=\lambda z,\;\alpha_2=\lambda(1-z),\;\alpha_3=(1-\lambda)\omega\;,\alpha_4=(1-\lambda)(1-\omega), \quad \lambda,z,\omega \in[0,1].
\end{eqnarray}
The integral then factorizes as
\begin{align}
    \left.\Phi\right|_{x=-1}&=s^{-2-\epsilon}\int_{0}^{1}\dd\lambda\:\left[\lambda(1-\lambda)\right]^{-1-\epsilon}\int_{0}^{1}\dd z\int_{0}^{1}\dd \omega\:(1-z-w-i0)^{-2-\epsilon}\notag\\
    &=\frac{-2}{\epsilon^2 s^2}\frac{\Gamma(1-\epsilon)^2}{(1+\epsilon)\Gamma(1-2\epsilon)}\left[(-s-i0)^{-\epsilon}+(s-i0)^{-\epsilon}\right]\notag\\
    &=\frac{-2\: (-s)^{-\epsilon}}{\epsilon^2 s^2}\frac{\Gamma(1-\epsilon)^2}{(1+\epsilon)\Gamma(1-2\epsilon)}\left[1+(-1-i0)^\epsilon\right].
\end{align}
Thus, from \eqref{eq:evalf}, we finally obtain that,
\begin{eqnarray}
    f(x=-1)=\frac{2}{\epsilon^2(1+\epsilon)}\frac{\Gamma(1-\epsilon)^2}{\Gamma(1-2\epsilon)}\left[1+(-1-i0)^\epsilon\right]=\frac{4}{\epsilon^2}+\mathcal{O}\left(\epsilon^{-1}\right).
\end{eqnarray}
Now, taking that $x=\frac{w_2w_4z_2}{w_1w_3z_1}=\frac{t}{s}$ and using \eqref{eq:evalf}, we can write
\begin{eqnarray}
    f(x)=x(-1-i0)^{\epsilon}\mathcal{I}(x)
\end{eqnarray}
where
\begin{eqnarray}
    \mathcal{I}(x)=\int_{0}^{1}\dd\alpha_1\dd\alpha_2\dd\alpha_3\dd\alpha_4\:\delta\left(1-\sum_{i=1}^{4}\alpha_i\right)(-\alpha_1\alpha_3-x\alpha_2\alpha_4-i0)^{-2-\epsilon}.
\end{eqnarray}
Hence,
\begin{eqnarray}
     f'(x)&=(-1-i0)^{\epsilon}\left[\mathcal{I}(x)+x\mathcal{I}'(x)\right].
\end{eqnarray}
Using the same change of variables introduced in \eqref{eq:newcoord}, we can finally evaluate the above expression at $x=-1$ and get 
\begin{equation}
    f'(x=-1)=\frac{2}{\epsilon(1-\epsilon^2)}\frac{\Gamma(1-\epsilon)^2}{\Gamma(1-2\epsilon)}\left[1-\epsilon(-1-i0)^\epsilon\right]=\frac{2}{\epsilon}+2+\mathcal{O}\left(\epsilon^{1}\right).
\end{equation}

\section{Explicit expressions for the Pfaffian matrices}\label{app:explicitGM}

For completeness, we will provide the explicit expressions for the matrices that appear in the examples.

\subsection{Triangle integral}
We start by presenting the components of the Pfaffian matrices in the original toric variables. For $P_x^{(0)}(x,y)$:
\begin{eqnarray}
&&P_{x11}^{(0)}=0 ,\quad P_{x12}^{(0)}=\frac{1}{x} ,\quad P_{x13}^{(0)}=0 ,\quad P_{x14}^{(0)}=0 ,\quad P_{x21}^{(0)}=-\frac{1}{x+y-1}\notag\\
&&P_{x22}^{(0)}=-\frac{2}{x+y-1} ,\quad P_{x23}^{(0)}=-\frac{2}{x+y-1} ,\quad P_{x24}^{(0)}=-\frac{2}{x+y-1} ,\quad P_{x31}^{(0)}=0 ,\quad P_{x32}^{(0)}=0\notag\\
&&P_{x33}^{(0)}=0 ,\quad P_{x34}^{(0)}=\frac{1}{x} ,\quad P_{x41}^{(0)}=-\frac{3 y (x-y+1)}{(x+y-1) \left(x^2-2 x (y+1)+(y-1)^2\right)}\notag\\
&&P_{x42}^{(0)}=\frac{2y(4(y-1)-2x)}{(x+y-1)\left(x^2-2x(y+1)+(y-1)^2\right)} ,\notag\\
&&P_{x34}^{(0)}=-\frac{x^2+2 x (4 y-1)+(-5 y-1) (y-1)}{(x+y-1) \left(x^2-2 x (y+1)+(y-1)^2\right)} ,\notag\\
&&P_{x44}^{(0)}=\frac{-2 x^3+x^2 (4-6 y)+x (y-1) (8 y+2)}{x (x+y-1) \left(x^2-2 x (y+1)+(y-1)^2\right)} ,\notag
\end{eqnarray}
next $P_x^{(1)}(x,y)$
\begin{eqnarray}
&&P_{x11}^{(1)}=0 ,\quad P_{x12}^{(1)}=0 ,\quad P_{x13}^{(1)}=0 ,\quad P_{x14}^{(1)}=0, \quad P_{x21}^{(1)}=-\frac{1}{x+y-1} ,\notag\\
&&P_{x22}^{(1)}=\frac{-x-y+1}{x (x+y-1)} ,\quad P_{x23}^{(1)}=0 ,\quad P_{x24}^{(1)}=0 ,\quad P_{x31}^{(1)}=0 ,\quad P_{x32}^{(1)}=0 ,\notag\\
&&P_{x33}^{(1)}=0 ,\quad P_{x34}^{(1)}=0 ,\quad P_{x41}^{(1)}=-\frac{3 y (x-y+1)}{(x+y-1) \left(x^2-2 x (y+1)+(y-1)^2\right)} ,\notag\\
&&P_{x42}^{(1)}=\frac{2 y (-x-y+1)}{(x+y-1) \left(x^2-2 x
   (y+1)+(y-1)^2\right)} ,\notag\\
&&P_{x34}^{(1)}=-\frac{-x^2-2 x y+2 x-y^2+2 y-1}{(x+y-1) \left(x^2-2 x (y+1)+(y-1)^2\right)} ,\notag\\
&&P_{x44}^{(1)}=\frac{x^3+x^2 y-x^2-x y^2+2 x y-x-y^3+3 y^2-3 y+1}{x (x+y-1) \left(x^2-2 x (y+1)+(y-1)^2\right)} ,\notag
\end{eqnarray}
following $P_y^{(0)}(x,y)$
\begin{eqnarray}
&&P_{y11}^{(0)}=0 ,\quad P_{y12}^{(0)}=0 ,\quad P_{y13}^{(0)}=\frac{1}{y} ,\quad P_{y14}^{(0)}=0 \quad P_{y21}^{(0)}=0 ,\quad P_{y22}^{(0)}=0 ,\notag\\
&&P_{y23}^{(0)}=0 ,\quad P_{y24}^{(0)}=\frac{1}{y} ,\quad P_{y31}^{(0)}=-\frac{1}{x+y-1} ,\quad P_{y32}^{(0)}=-\frac{2}{x+y-1} ,\notag\\
&&P_{y33}^{(0)}=-\frac{2}{x+y-1} ,\quad P_{y34}^{(0)}=-\frac{2}{x+y-1} ,\notag\\
&&P_{y41}^{(0)}=\frac{3 x (x-y-1)}{(x+y-1) \left(x^2-2 x (y+1)+(y-1)^2\right)} ,\notag\\
&&P_{y42}^{(0)}=-\frac{-5x^2+2x(4y+2)+(y-1)^2}{(x+y-1)\left(x^2-2x(y+1)+(y-1)^2\right)} ,\notag\\
&&P_{y34}^{(0)}=\frac{2 x (4 x-2 y-4)}{(x+y-1) \left(x^2-2 x (y+1)+(y-1)^2\right)} ,\notag\\
&&P_{y44}^{(0)}=\frac{8 x^2 y-x \left(6 y^2+6 y\right)-2 (y-1)^2 y}{y (x+y-1) \left(x^2-2 x (y+1)+(y-1)^2\right)} ,\notag
\end{eqnarray}
finally $P_y^{(1)}(x,y)$
\begin{eqnarray}
&&P_{y11}^{(2)}=0 ,\quad P_{y12}^{(2)}=0 ,\quad P_{y13}^{(2)}=0 ,\quad P_{y14}^{(2)}=0 ,\quad P_{y21}^{(2)}=0 ,\quad P_{y22}^{(2)}=0 ,\notag\\
&&P_{y23}^{(2)}=0 ,\quad P_{y24}^{(2)}=0 ,\quad P_{y31}^{(2)}=-\frac{1}{x+y-1} ,\quad P_{y32}^{(2)}=0 \quad, P_{y33}^{(2)}=\frac{-x-y+1}{y (x+y-1)} ,\notag\\
&&P_{y34}^{(2)}=0 ,\quad P_{y41}^{(2)}=\frac{3 x (x-y-1)}{(x+y-1) \left(x^2-2 x (y+1)+(y-1)^2\right)} ,\notag\\
&&P_{y42}^{(2)}=-\frac{-x^2-2 x y+2 x-y^2+2 y-1}{(x+y-1)\left(x^2-2x(y+1)+(y-1)^2\right)} ,\notag\\
&&P_{y34}^{(2)}=\frac{2 x (-x-y+1)}{(x+y-1) \left(x^2-2 x (y+1)+(y-1)^2\right)} ,\notag\\
&&P_{y44}^{(2)}=\frac{-x^3-x^2 y+3 x^2+x y^2+2 x y-3 x+y^3-y^2-y+1}{y(x+y-1)\left(x^2-2x(y+1)+(y-1)^2\right)} .\notag
\end{eqnarray}

The components for the gauge transformation matrices will be written in components due to their sizes. Starting with $T^{(0)}(u,v)$:
\begin{eqnarray}
&&T_{11}^{(0)}=\frac{v+1}{u v+u-1} ,\quad T_{12}^{(0)}=0 ,\quad T_{13}^{(0)}=0 ,\quad T_{14}^{(0)}=0 ,\notag\\
&&T_{21}^{(0)}=-\frac{u (v+1) ((u-2) v+u-1)}{(u v+u-1)^3} ,\quad T_{22}^{(0)}=0 ,\quad T_{23}^{(0)}=0 ,\quad T_{24}^{(0)}=0 ,\notag\\
&&T_{31}^{(0)}=-\frac{(u-1) v (v+1) (u v+u+1)}{(u v+u-1)^3} ,\quad T_{32}^{(0)}=0 ,\quad T_{33}^{(0)}=0 ,\quad T_{34}^{(0)}=0 ,\notag\\
&&T_{41}^{(0)}=\frac{2 (u-1) u v (v+1) (u (v+1) ((u-3) v+u+1)-3 v-2)}{(u v+u-1)^5} ,\quad T_{42}^{(0)}=0 ,\quad T_{43}^{(0)}=0 ,\notag\\
&&T_{44}^{(0)}=0 ,\notag
\end{eqnarray}
then $T^{(1)}(u,v)$:
\begin{eqnarray}
&&T_{11}^{(1)}=0 ,\quad T_{12}^{(1)}=0 ,\quad T_{13}^{(1)}=0 ,\quad T_{14}^{(1)}=0 ,\notag\\
&&T_{21}^{(1)}=\frac{2 (u-1) (v+1) (u v+u+1)}{(u v+u-1)^3} ,\quad T_{22}^{(1)}=\frac{2 u (v+1)}{(u v+u-1)^2} ,\quad T_{23}^{(1)}=\frac{(v+1) (u (v+3)+1)}{(uv+u-1)^2} ,\notag\\
&&T_{24}^{(1)}=0 ,\quad T_{31}^{(1)}= -\frac{2 (u-1) (v+1) ((u-2) v+u-1)}{(u v+u-1)^3} ,\quad T_{32}^{(1)}=\frac{u \left(v^2-1\right)+v+1}{(u v+u-1)^2} ,\notag\\
&&T_{33}^{(1)}=-\frac{2(v+1) (u-v-1)}{(u v+u-1)^2} ,\quad T_{34}^{(1)}=0 ,\notag\\
&&T_{41}^{(1)}=-\frac{2 (u-1) (v+1) \left(u^2 (2 u-3) v^3+3 u (u (u+2)-4) v^2+(u-1) (11 u+1) v-(u-1)^2 u\right)}{(u v+u-1)^5} \notag\\
&&T_{42}^{(1)}= -\frac{u (v+1) \left(u \left(u (v+1) (5 v-1)-6 v^2+2 v+2\right)-6 v-1\right)}{(u v+u-1)^4},\notag\\
&&T_{43}^{(1)}=-\frac{(v+1)
   \left(u \left(u^2 (v+1) (v (v+5)-2)+u (v (13-(v-2) v)+4)-5 v (2 v+3)-2\right)-v\right)}{(u v+u-1)^4} ,\notag\\
&&T_{44}^{(1)}=0 ,\notag
\end{eqnarray}
finally $T^{(2)}(u,v)$:
\begin{eqnarray}
&&T_{11}^{(2)}=0 ,\quad T_{12}^{(2)}=0 ,\quad T_{13}^{(2)}=0 ,\quad T_{14}^{(2)}=0 ,\notag\\
&&T_{21}^{(2)}=0 ,\quad T_{22}^{(2)}=0 ,\quad T_{23}^{(2)}=0 ,\quad T_{24}^{(2)}=0 ,\notag\\
&&T_{31}^{(2)}=0 ,\quad T_{32}^{(2)}=0 ,\quad T_{33}^{(2)}=0 ,\quad T_{34}^{(2)}=0 ,\notag\\
&&T_{41}^{(2)}=\frac{2 (u-1) (v+1) (u (v-1)+1) ((u-2) v+u-1) (u v+u+1)}{(u v+u-1)^5} ,\notag\\
&&T_{42}^{(2)}=\frac{2 u (v+1) (u (v-1)+1) ((u-2)v+u-1)}{(u v+u-1)^4} ,\notag\\
&&T_{43}^{(2)}=\frac{2 (v+1) (u (v-1)+1) (u-v-1) (u v+u+1)}{(u v+u-1)^4} ,\quad T_{44}^{(2)}=\frac{u\left(v^2-1\right)+v+1}{(u v+u-1)^2} .\notag
\end{eqnarray}

\subsection{Sunset integral}
Pfaffian matrices in the original toric variables
\begin{eqnarray}
 P^{(0)}_x(x,y)&=&
\begin{pmatrix}
 0 & \frac{1}{x} & 0 \\
 -\frac{2 x+2 y^2-2 y}{(4 x-1) (x+(y-1) y)} & \frac{x \left(6 y-6 y^2\right)-6 x^2}{x (4 x-1) (x+(y-1) y)} &
   -\frac{(y-1) (2 y-1)}{(4 x-1) (x+(y-1) y)} \\
 0 & 0 & -\frac{1}{x+(y-1) y}
\end{pmatrix},\notag\\
 P^{(1)}_x(x,y)&=&
\begin{pmatrix}
 0 & 0 & 0 \\
 \frac{2 x y+3 x+4 y^2-4 y}{(4 x-1) \left(x+y^2-y\right)} & \frac{x y+x+2 y^2-2 y}{x \left(x+y^2-y\right)} &
   \frac{x (2 y+1)+2 (y-1) y}{(4 x-1) (x+(y-1) y)} \\
 -\frac{y}{x+(y-1) y} & -\frac{y (2 x+y-1)}{x (x+(y-1) y)} & -\frac{y}{x+(y-1) y}
\end{pmatrix},\notag\\
 P^{(2)}_x(x,y)&=&
\begin{pmatrix}
0 & 0 & 0 \\
 \frac{-2 (x-1) y-x-2 y^2}{(4 x-1) (x+(y-1) y)} & 0 & 0 \\
 \frac{y}{x+(y-1) y} & 0 & 0 
\end{pmatrix},\notag\\
 P^{(0)}_y(x,y)&=&
\begin{pmatrix}
0 & 0 & \frac{1}{y} \\
 0 & 0 & -\frac{x}{y (x+(y-1) y)} \\
 0 & 0 & \frac{x-y^2}{y (x+(y-1) y)}
\end{pmatrix},\notag\\
 P^{(1)}_y(x,y)&=&
\begin{pmatrix}
0 & 0 & 0 \\
 -\frac{x}{x+(y-1) y} & \frac{-2 x-y+1}{x+(y-1) y} & -\frac{x}{x+(y-1) y} \\
 \frac{2 x-y}{x+(y-1) y} & \frac{4 x-1}{x+(y-1) y} & \frac{2 x y+x-y}{y (x+(y-1) y)} 
\end{pmatrix},\notag\\
 P^{(2)}_y(x,y)&=&
\begin{pmatrix}
0 & 0 & 0 \\
 \frac{x}{x+(y-1) y} & 0 & 0 \\
 \frac{y-2 x}{x+(y-1) y} & 0 & 0
\end{pmatrix}.\notag
\end{eqnarray}

The gauge transformation matrices will be written in components due to their sizes. Starting with $T^{(0)}(u,v)$:
\begin{eqnarray}
&&T_{11}^{(0)}=-\frac{u+v}{4 (u-v)},\quad T_{12}^{(0)}=0,\quad T_{13}^{(0)}=0 \notag\\
&&T_{21}^{(0)}= \frac{v (u+v) \left(-2 u^5+4 u^4+u^3 (4 v-2)-4 u^2 v-2 u v^2\right)}{4 (u-v)^3 \left(-u^2+u+v\right)^2},\quad T_{22}^{(0)}=0,\quad
T_{23}^{(0)}=0 \notag\\
&&T_{31}^{(0)}= 0,\quad T_{32}^{(0)}=0,\quad T_{33}^{(0)}=0,\notag
\end{eqnarray}
then $T^{(1)}(u,v)$:
\begin{eqnarray}
&&T_{11}^{(1)}=\frac{u+v}{4 u-4 v},\quad T_{12}^{(1)}=0,\quad T_{13}^{(1)}=0 \notag\\
&&T_{21}^{(1)}=\frac{v (u+v) \left(3 u^5-2 u^4 (v+3)+u^3 (v-2)^2+2 u^2 v (v+3)-2 v^3\right)}{4 (u-v)^3
   \left(-u^2+u+v\right)^2},\quad T_{22}^{(1)}=-\frac{(u+v)^2}{4 (u-v)^2},\notag\\ 
&&T_{23}^{(1)}=-\Big((u+v)^2 \big(-2 (u-1)^2 u^4+\left(u\left(u^3-6 u+12\right)-8\right) u v^3+\left(u \left(u^3-6 u+12\right)-8\right) u^3 v \notag\\
&& \quad-2 \left((u-2) u\left(u^2+4\right)+6\right) u^2 v^2-2 (u-1)^2 v^4\big)\Big)\frac{1}{(u-v)^2 \left(-u^2+u+v\right)^2 (u (-v)+u+v)^2}, \notag\\
&&T_{31}^{(1)}=-\frac{(u+v) \left(\left(u^2-2\right) u v-u^2-v^2\right)}{4 (u-v) \left(-u^2+u+v\right)^2},\quad T_{32}^{(1)}=0,\notag\\
&&T_{33}^{(1)}=\frac{(u-2) u(u+v)^2 \left(\left(u^2-2\right) u v-u^2-v^2\right)}{\left(-u^2+u+v\right)^2 (u (-v)+u+v)^2},\notag
\end{eqnarray}
finally $T^{(2)}(u,v)$:
\begin{eqnarray}
&&T_{11}^{(2)}=0,\quad T_{12}^{(2)}=0,\quad T_{13}^{(2)}=0 \notag\\
&&T_{21}^{(2)}=-\frac{v (u+v) \left(((u-2) u+2) u^2-(u (u+2)-4) u v+2 v^2\right)}{4 (u-v)^2 \left(-u^2+u+v\right)^2},\quad T_{22}^{(2)}=\frac{(u+v)^2}{4 (u-v)^2},\notag\\
&&T_{23}^{(2)}=\Big((u+v)^2 \big(-2 (u-1)^2 u^4+\left(u \left(u^3-6 u+12\right)-8\right) uv^3+\left(u \left(u^3-6 u+12\right)-8\right) u^3 v \notag\\
&& \quad-2 \left((u-2) u \left(u^2+4\right)+6\right) u^2 v^2-2
   (u-1)^2 v^4\big)\Big)\frac{1}{(u-v)^2 \left(-u^2+u+v\right)^2 (u (-v)+u+v)^2},\notag\\
&&T_{31}^{(2)}=\frac{(u+v) \left(\left(u^2-2\right) u v-u^2-v^2\right)}{4 (u-v) \left(-u^2+u+v\right)^2},\quad T_{32}^{(2)}=0,\notag\\
&&T_{33}^{(2)}=-\frac{(u-2)u(u+v)^2 \left(\left(u^2-2\right) u v-u^2-v^2\right)}{\left(-u^2+u+v\right)^2 (u (-v)+u+v)^2}.\notag
\end{eqnarray}

\providecommand{\href}[2]{#2}\begingroup\raggedright\endgroup

\end{document}